\documentclass[a4paper,USenglish,cleveref,autoref,thm-restate]{lipics-v2021}
\hideLIPIcs
\nolinenumbers

\PassOptionsToPackage{full}{optional}

\title{Differentially Private Algorithms for Graphs Under Continual Observation}

\author{Hendrik Fichtenberger}{University of Vienna, Austria}{hendrik.fichtenberger@univie.ac.at}{https://orcid.org/0000-0003-3246-5323}{}
\author{Monika Henzinger}{University of Vienna, Austria}{monika.henzinger@univie.ac.at}{https://orcid.org/0000-0002-5008-6530}{}
\author{Lara Ost}{University of Vienna, Austria}{lara.ost@univie.ac.at}{}{}

\authorrunning{H. Fichtenberger and M. Henzinger and W. Ost}

\Copyright{Hendrik Fichtenberger and Monika Henzinger and Wolfgang Ost}



\acknowledgements{Research supported by the Austrian Science Fund (FWF)
and netIDEE SCIENCE project P 33775-N.
Lara Ost gratefully acknowledges the financial support from the Vienna Graduate School on Computational Optimization which is funded by the Austrian Science Fund (FWF, project no. W1260-N35).}

\EventEditors{John Q. Open and Joan R. Access}
\EventNoEds{2}
\EventLongTitle{42nd Conference on Very Important Topics (CVIT 2016)}
\EventShortTitle{CVIT 2016}
\EventAcronym{CVIT}
\EventYear{2016}
\EventDate{December 24--27, 2016}
\EventLocation{Little Whinging, United Kingdom}
\EventLogo{}
\SeriesVolume{42}
\ArticleNo{23}
\EventEditors{Petra Mutzel, Rasmus Pagh, and Grzegorz Herman}
\EventNoEds{3}
\EventLongTitle{29th Annual European Symposium on Algorithms (ESA 2021)}
\EventShortTitle{ESA 2021}
\EventAcronym{ESA}
\EventYear{2021}
\EventDate{September 6--8, 2021}
\EventLocation{Lisbon, Portugal}
\EventLogo{}
\SeriesVolume{204}
\ArticleNo{15}

\usepackage{optional}
\usepackage{tagging}

\opt{conf}{
\relatedversion{}
\relatedversiondetails{Full Version}{https://arxiv.org/abs/2106.14756}
}
  
\usepackage{amsmath,amssymb,amsthm}
\allowdisplaybreaks
\usepackage[noend,ruled,linesnumbered]{algorithm2e}
\SetKwProg{Fn}{Function}{}{}
\DontPrintSemicolon
\usepackage{booktabs}

\usepackage{graphicx}
\usepackage{float}

\usepackage{tikz}
\usetikzlibrary{patterns,fit}
\pgfdeclarelayer{bg}
\pgfsetlayers{bg,main}

\usepackage{tabulary}
\usepackage{makecell}

\usepackage{xparse}
\usepackage{xcolor}
\definecolor{col1}{RGB}{213, 94,  0}
\definecolor{col2}{RGB}{  0,158,115}
\definecolor{col3}{RGB}{  0,114,178}
\definecolor{col4}{RGB}{213, 94,178}

\newtheorem{fact}[theorem]{Fact}

\newcommand{\Vdel}[1]{\partial {V^-_{#1}}}
\newcommand{\Vins}[1]{\partial {V^+_{#1}}}

\newcommand{\Edel}[1]{\partial {E^-_{#1}}}
\newcommand{\Eins}[1]{\partial {E^+_{#1}}}
\newcommand{\graphseq}{\mathcal{G}}     \newcommand{\algo}{\mathcal{A}}         \newcommand{\range}{\mathrm{Range}}     \newcommand{\eps}{\varepsilon}          \newcommand{\Lap}{\mathrm{Lap}}         \newcommand{\reals}{\mathbb{R}}         \newcommand{\nats}{\mathbb{N}}          \newcommand{\integers}{\mathbb{Z}}      \newcommand{\TT}{^\top}                                      \newcommand{\psstart}{\mathrm{start}}
\newcommand{\psend}{\mathrm{end}}
\newcommand{\GS}{\operatorname{GS}}     \newcommand{\GSstatic}{\operatorname{GS_{static}}}     \newcommand{\bin}{\operatorname{bin}}   

\newcommand{\wmst}{{w_{\mathrm{MST}}}}    \newcommand{\wcut}{{w_{\mathrm{CUT}}}}    \newcommand{\wmatch}{{w_{\mathrm{M}}}}    

\newcommand{\flfrac}[2]{\left\lfloor\frac{#1}{#2}\right\rfloor}  

\NewDocumentCommand{\apply}{m m}{#1 \oplus #2}

\makeatletter
\newcommand{\etal}{et al\@ifnextchar.{}{.\@}}
\makeatother

\makeatletter
\newcommand{\afoot}{\textsuperscript{a}}
\makeatother

\opt{conf}{}

\begin{document}
\maketitle

\begin{abstract}
 Differentially private algorithms protect individuals in data analysis scenarios by ensuring that there is only a weak correlation between the existence of the user in the data and the result of the analysis. Dynamic graph algorithms maintain the solution to a problem (e.g., a matching) on an evolving input, i.e., a graph where nodes or edges are inserted or deleted over time. They output the value of the solution after each update operation, i.e., continuously. We study (event-level and user-level) differentially private algorithms for graph problems under continual observation, i.e., differentially private dynamic graph algorithms.
 We present event-level private algorithms for partially dynamic counting-based problems such as triangle count that improve the additive error by a polynomial factor (in the length $T$ of the update sequence)  on the state of the art, resulting in the first algorithms with additive error polylogarithmic in $T$.

 We also give $\eps$-differentially private and partially dynamic algorithms for minimum spanning tree, minimum cut, densest subgraph, and maximum matching. The additive error of our improved MST algorithm is $O(W \log^{3/2}T / \eps)$, where $W$ is the maximum weight of any edge, which, as we show, is tight up to a $(\sqrt{\log T} / \eps)$-factor. For the other problems, we present a partially-dynamic algorithm with multiplicative error $(1+\beta)$ for any constant $\beta > 0$ and additive error $O(W \log(nW) \log(T) / (\eps \beta) )$. Finally, we show that the additive error for a broad class of dynamic graph algorithms with user-level privacy must be linear in the value of the output solution's range.
\end{abstract}

\section{Introduction}

Differential privacy aims to protect individuals whose data becomes part of an increasing number of data sets and is subject to analysis. A differentially private algorithm guarantees that its output depends only very little on an individual's contribution to the input data. Roughly speaking, an algorithm is \emph{$\epsilon$-differentially private} if the probability that it outputs $O$ on data set $\mathcal{D}$ is at most an $e^\epsilon$-factor of the probability that it outputs $O$ on any adjacent data set $\mathcal{D}'$. Two data sets are \emph{adjacent} if they differ only in the data of a single user.
Differential privacy was introduced in the setting of databases \cite{DwoDif06,DwoCal06}, where users (entities) are typically represented by rows and data is recorded in columns. An important notion that allowed for the development of generic techniques and tools (like the Laplace and the exponential mechanism) is the \emph{sensitivity} of a function $f$: the static sensitivity $\rho$ of $f$ is the maximum $\lvert f(\mathcal D) - f(\mathcal D') \rvert$ over all adjacent pairs $\mathcal{D},\mathcal{D}'$. Differential privacy was later generalized to a more challenging setting, where data evolves over time~\cite{dwork10,chan11}: a differentially private algorithm under \emph{continual observation} must provide the same privacy guarantees as before, but for a sequence (or stream) of data sets instead of just a single data set. Often, this sequence results from updates to the original data set that arrive over time. In this setting, the presence or absence of a single user in one update can affect the algorithm's output on all future data sets, i.e., two adjacent databases can differ on all future outputs and, thus, have infinite sensitivity.

In this paper, we study differentially private graph algorithms under continual observation, i.e., for \emph{dynamic} graph problems. The input is a sequence of graphs that results from node or edge updates, i.e., insertions or deletions. \emph{Partially dynamic} algorithms only allow either insertions or deletions, \emph{fully dynamic} algorithms allow both. After each update, the algorithm has to output a solution for the current input, i.e., the algorithm outputs a sequence of answers that is equally long as the input sequence.
For differentially private graph algorithms two notions of \emph{adjacency of graph sequences} exist: node-adjacency and edge-adjacency. Two graph sequences are \emph{edge-adjacent} if they only differ in a single insertion or deletion of an edge. Similarly, two graph sequences are \emph{node-adjacent} if they only differ in an insertion or deletion of a node.\footnote{Of course, a graph can also be represented by a database, where, e.g., every row corresponds to an edge, but as we present algorithms that solve graph algorithmic problems we use the graph-based terminology through the paper.}
Some of our algorithms assume that they are given as input an upper bound on some graph parameter, such as   the maximum degree $D$ or the maximum edge weight $W$ in the graph and that it is guaranteed that the input graph respects these bounds. Our algorithms do not check whether this guarantee is indeed respected by the input sequence. If no bound on the maximum degree is given, the algorithms assume that $D$ equals the trivial upper bound of $n$.

We initiate the study of differentially private algorithms for \emph{non-local} partially dynamic graph problems. 
We consider a problem \emph{non-local} if its (optimum) value cannot be derived from the frequency histogram of constant-size subgraphs of the input graph and call it \emph{local} otherwise.
Non-local problems include the cost of the minimum spanning tree, the weight of the global and $s$-$t$ minimum cut, and the density of the densest subgraph. We also give improved algorithms for local graph problems and show various lower bounds on the additive error for differentially private dynamic graph algorithms.

\begin{table}
\caption{Additive errors for partially-dynamic $\eps$-differentially private algorithms with failure probability $\delta$. We use $D$ for the maximum degree and $W$ for the maximum edge weight, $n$ for the maximum number of nodes of any graph in the input sequence, and $\Lambda = \log(1/\delta)/\eps$.  All these are publicly known parameters. The upper bounds follow from \cref{cor:graph-bin-mech} and \cref{tbl:song-sensitivity}\opt{conf}{ on page~\pageref{tbl:song-sensitivity}}.
See \cref{sec:edge-dp-lower-bounds} for results on event-level lower bounds and \cref{sec:lower-bound-user-level} for user-level lower bounds.}
\label{tbl:sensitivity-results}
\centering
\begin{tabulary}{\linewidth}{LRRRR}
        \toprule
        Graph function            & \multicolumn{2}{c}{partially dynamic} & \multicolumn{2}{c}{fully dynamic} \\
                                  & edge-adj.      & node-adj.            & edge-adj.             & edge-adj. \\
                                  & event-level    & event-level          & event-level           & user-level    \\
        \midrule
        min. spanning tree     & \makecell[tr]{$\Omega(W\log T)$, \\ $O(W\log^{3/2} T \cdot \Lambda)$}
                               & \makecell[tr]{$\Omega(W\log T)$, \\ $O(DW\log^{3/2} T \cdot \Lambda)$}
                               & $\Omega(W\log T)$                                             & $\Omega(nW)$ \\
				\midrule
\makecell[tl]{min. cut,\\ max. matching}            & $\Omega(W\log T)$ & $\Omega(W\log T)$ & $\Omega(W\log T)$     & $\Omega(nW)$ \\
					\midrule
        edge count             & \makecell[tr]{$\Omega(\log T)$, \\ $O(\log^{3/2}T \cdot \Lambda)$}
                               & \makecell[tr]{$\Omega(D \log T)$, \\ $O(D\log^{3/2}T \cdot \Lambda)$}
                               & \makecell[tr]{$\Omega(\log T)$, \\ $O(\log^{3/2}T \cdot \Lambda)$}  & $\Omega(n^2)$ \\
				\midrule
        high-degree nodes      & \makecell[tr]{$\Omega(\log T)$, \\ $O(\log^{3/2}T \cdot \Lambda))$}
                               & \makecell[tr]{$\Omega(D \log T)$, \\ $O(D\log^{3/2}T \cdot \Lambda)$}
                               & $\Omega(\log T)$                                                 & $\Omega(n)$ \\
					\midrule
        degree histogram       & \makecell[tr]{$\Omega(\log T)$, \\ $O(D\log^{3/2}T \cdot \Lambda))$}
                               & \makecell[tr]{$\Omega(D \log T)$, \\ $O(D^2\log^{3/2}T \cdot \Lambda)$}
                               & $\Omega(\log T)$                                                 & $\Omega(n)$ \\
					\midrule
        triangle count         & \makecell[tr]{$\Omega(\log T)$, \\ $O(D\log^{3/2}T \cdot \Lambda))$}
                               & \makecell[tr]{$\Omega(D \log T)$, \\ $O(D^2\log^{3/2}T \cdot \Lambda)$}
                               & $\Omega(\log T)$                                                 & $\Omega(n^3)$ \\
						\midrule
        $k$-star count         & \makecell[tr]{$\Omega(\log T)$, \\ $O(D^k\log^{3/2}T \cdot \Lambda)$}
                               & \makecell[tr]{$\Omega(D \log T)$, \\ $O(D^k\log^{3/2}T \cdot \Lambda)$}
                               & $\Omega(\log T)$                                                 & $\Omega(n^{k+1})$ \\
        \bottomrule
\end{tabulary}
\end{table}

{\bf Local problems.}
The only prior work on differentially private dynamic algorithms is an algorithm by Song et al.~~\cite{song18} for various \emph{local} graph problems such as counting high-degree nodes, triangles and other constant-size subgraphs.  Even though not explicitly stated, they make the same assumption with respect to the maximum degree $D$ as our algorithms.
We present an algorithm  for these local problems  that improves the additive error by a factor of $\sqrt{T}/\log^{3/2}T$, where $T$ is the length of the update sequence.
We also give the first differentially private partially-dynamic algorithm for the value of the minimum spanning tree. 
\Cref{tbl:sensitivity-results} lists upper and lower bounds for these results, where $n$ is the number of nodes in the graph, $W$ is the maximum edge weight (if applicable), $D$ is the maximum node degree, $\epsilon$ is an arbitrarily small positive constant, and $\delta$ is the failure probability of the algorithm. 
We state below our main contributions in more detail. 
The update time of all our algorithms is linear in $\log T$  plus the time needed to solve the corresponding non-differentially private dynamic graph problem.

\begin{theorem}[see \cref{sec:global-sens-mech}]\label{thm:1}
Let $\eps, \delta > 0$.
        There exist an $\eps$-edge-differentially as well as an $\eps$-node-differentially private algorithm for partially-dynamic minimum spanning tree, edge count, the number of high-degree nodes, the degree histogram, triangle count and $k$-star count that with
				probability at least $1-\delta$ give an answer with additive error as shown in \Cref{tbl:sensitivity-results}. {The algorithms whose accuracy bound in that table contains the parameter $D$ assume that $D$ is a publicly known upper bound on the maximum degree of the input graph. The other algorithms do not need this assumption.}
\end{theorem}

{\bf Non-local problems.}
For non-local problems we present an algorithm that, by allowing a small multiplicative error, can obtain differentially private partially dynamic algorithms for a broad class of problems that includes the aforementioned problems. \Cref{tbl:monotone-results} lists our results for some common graph problems.
The algorithm achieves the following performance.

\begin{theorem}[see \cref{thm:technical-monotone}]\label{thm:2}
        Let $\eps, \beta, \delta, r > 0$ and let $f$ be a function with range $[1,r]$ that is monotone on all input sequences and has static global sensitivity\footnote{Defined formally in~\cref{subsec:defs}.} $\rho$, where $r$ and $\rho$ are publicly known parameters. There exists an $\eps$-differentially private dynamic algorithm with multiplicative error $(1+\beta)$, additive error $O(\rho \log(r) \log(T) / \log(1+\delta))$ and failure probability $\delta$ that computes $f$. 
\end{theorem}
Note that for partially dynamic graph algorithms it holds that $T = O(n^2)$. 

{\bf Lower bounds.}
We complement these upper bounds by also giving some lower bounds on the additive error of any differentially private dynamic graph algorithm. For the problems in~\Cref{tbl:sensitivity-results} we show lower bounds
of $\Omega(W \log T)$, resp.~$\Omega(\log T)$, where $W$ is assumed to be a publically known maximum allowed edge weight. Note that these lower bounds apply to the partially dynamic as well as to the fully dynamic setting.

The above notion of differential privacy is also known as \emph{event-level} differential privacy, where two graph sequences differ in at most one ``event'', i.e., one update operation. A more challenging notion is \emph{user-level} differential privacy.
 Two graph sequences are edge-adjacent \emph{on user-level} if they differ in \emph{any} number of updates for a single edge (as opposed to \emph{one} update for a single edge in the case of the former \emph{event-level} adjacency). Note that requiring user-level edge-differential privacy is a more stringent requirement on the algorithm than event-level edge-differential privacy.\footnote{Node-adjacency on user-level is defined accordingly but not studied in this paper.}
We show strong lower bounds for edge-differentially private algorithms on user-level for a broad class of dynamic graph problems.

\begin{theorem}[informal, see \cref{thm:technical-lower-bound}]
        Let $f$ be a function on graphs, and let $G_1, G_2$ be arbitrary graphs. There exists a $T \geq 1$ so that any $\eps$-edge-differentially private dynamic algorithm on user-level that computes $f$ must have additive error $\Omega(\lvert f(G_1) - f(G_2) \rvert)$ on input sequences of length at least $T$.
\end{theorem}
This theorem leads to the lower bounds for fully dynamic algorithms stated in \Cref{tbl:sensitivity-results}.

\begin{table}
\caption{Private algorithms with failure probability $\delta$ with additional multiplicative error of $(1+\beta)$ for arbitrary $\beta > 0$. We use $\Lambda = 1/(\eps \delta\log(1+\beta))$, $D$ for the maximum degree, $W$ for the maximum edge weight and $n$ for the maximum number of nodes of any graph in the input sequence. All these are publicly known parameters.}
\label{tbl:monotone-results}
\centering
\begin{tabulary}{\linewidth}{LRR}
        \toprule
        Graph function    & \multicolumn{2}{c}{partially dynamic, event-level} \\
                          & edge-adjacency & node-adjacency       \\
        \midrule
        minimum cut       & $O(W\log(nW) \log(T)\cdot \Lambda)$ & $O(DW\log(nW) \log(T)\cdot \Lambda)$ \\
        densest subgraph  & $O(\log(n) \log(T)\cdot \Lambda)$ & $O(\log(n) \log(T)\cdot \Lambda)$ \\
        minimum $s,t$-cut & $O(W\log(nW) \log(T)\cdot \Lambda)$ & $O(DW\log(nW) \log(T)\cdot \Lambda)$ \\
        maximum matching  & $O(W\log(nW) \log(T)\cdot \Lambda)$ & $O(W\log(nW) \log(T)\cdot \Lambda)$ \\
        \bottomrule
\end{tabulary}
\end{table}

\textbf{\sffamily Technical contribution.} 
{\bf Local problems.} Our algorithms for local problems (Theorem~\ref{thm:1}) incorporate a counting scheme by Chan et al.~\cite{chan11} and the \emph{difference sequence technique} by Song et al.~\cite{song18}. The difference sequence technique addresses the problem that two adjacent graphs might
differ on all outputs starting from the point in the update sequence where their inputs differ. 
More formally, let $f_{\graphseq}(t)$ be the output of the algorithm after operation $t$  in the  graph sequence $\graphseq$.
Then the \emph{continuous global sensitivity} $\sum_{t=1}^T |f_{\graphseq}(t) - f_{\graphseq'}((t)|$ might be $\Theta(\rho T)$. Using the ``standard'' Laplacian mechanism for such a large sensitivity would, thus, lead to an additive error linear in $T$. 
The idea of~\cite{song18} is to use instead the \emph{difference sequence of $f$} defined as $\Delta f(t) = f(t) - f(t-1)$, as they observed that for various local graph properties  the continuous global sensitivity of the difference sequence, i.e., $\sum_{t=1}^T |\Delta f_{\graphseq}(t) - \Delta f_{\graphseq'}((t)|$ can be bounded by a function independent of $T$. 
However, their resulting partially-dynamic algorithms still have an additive error linear in $\sqrt T$.
We show how to combine the continuous global sensitivity of the difference sequence with the binary counting scheme of
Chan et al.~\cite{chan11} to achieve partially-dynamic algorithms with additive error linear in $\log^{3/2}T$.

Furthermore we show that the approach based on the continuous global sensitivity of the difference sequence fails, if the presence or absence of a node or edge can significantly change the target function's value for all of the subsequent graphs. 
In particular, we show that for several graph problems like minimum cut and maximum matching changes in the function value between adjacent graph sequences can occur at every time step even for partially dynamic sequences, resulting in a continuous global sensitivity of the difference sequence that is linear in $T$. This implies that this technique cannot be used to achieve differentially private dynamic algorithms for these problems.

{\bf Non-local problems.} We leverage the fact that the sparse vector technique~\cite{DwoCom09} provides negative answers to threshold queries with little effect on the additive error to approximate monotone functions $f$ on graphs under continual observation (e.g., the minimum cut value in an incremental graph) with multiplicative error $(1 +\beta)$: If $r$ is the maximum value of $f$, we choose thresholds $(1+\beta), \ldots, (1+\beta)^{\log_{1+\beta}(r)}$ for the queries. This results in at most $\log_{1+\beta}(r)$ positive answers, which affect the additive error linearly, while the at most $T$ negative answers affect the additive error only logarithmically instead of linearly.

{\bf Lower bounds.} Dwork et al.~\cite{dwork10} had given a lower bound for counting in binary streams. We reduce this problem to partially dynamic graph problems on the event-level to achieve the event-level lower bounds.

For the user-level lower bounds we assume by contradiction that an $\epsilon$-differentially private dynamic algorithm $\cal A$ with ``small'' additive error exists and
construct an exponential number of graph sequences that are all user-level ``edge-close'' to a simple graph sequence $\graphseq'$. Furthermore any two such graph sequences have at least one position with two very different graphs such that $\cal A$ (due to its small additive error) must return two different outputs at this position, which leads to two different output sequences if $\cal A$ answers within its error bound.
Let $O_i$ be the set of accurate output sequences of $\cal A$
on one of the graph sequences $G_i$. By the previous condition  $O_i \cap O_j = \emptyset$ if $i \not= j$.
As $G_i$ is ``edge-close'' to $\graphseq'$, 
there is a relatively large probability (depending on the degree of ``closeness'') that $O_i$ is output when $\cal A$ runs on $\graphseq'$. This holds for all $i$. 
However, since  $O_i \cap O_j = \emptyset$ if $i \not= j$ and we have constructed exponentially many graph sequences $G_i$, the sum of these probabilities over all $i$ adds up to a value larger than 1, which gives a contradiction.
The proof is based on ideas of a lower bound proof for databases in~\cite{dwork10}.

\opt{conf}{All missing proofs can be found in the full version of the paper at http://arxiv.org/abs/2106.14756.}

\textbf{\sffamily Related Work.} Differential privacy, developed in~\cite{DwoDif06,DwoCal06}, is the de facto gold standard of privacy definitions and several lines of research have since been investigated~\cite{BarPri07,McSMec07,KasWha08,DwoDif09,GupDif10,BluLea13}. In particular, differentially private algorithms for the release of various graph statistics such as subgraph counts~\cite{KarPri11,BloDif13,CheRec13,KasAna13,ZhaPri15}, degree distributions~\cite{HayAcc09,DayPub16,ZhaDif21}, minimum spanning tree~\cite{NisSmo07}, spectral properties \cite{WanDif13,AroDif19}, cut problems~\cite{GupDif10,AroDif19,EliDif19}, and parameter estimation for special classes of graphs~\cite{LuExp14} have been proposed. Dwork et al.~\cite{dwork10} and Chan et al.~\cite{chan11} extended the analysis of differentially private algorithms to the regime of continual observation, i.e., to input that evolves over time. Since many data sets in applications are evolving data sets, this has lead to results for several problems motivated by practice~\cite{ChaDif12,KelDif14,NyDif14,ErdPri15,WanRea18}. Only one prior work analyzes evolving graphs: Song et al.~\cite{song18} study problems in incremental bounded-degree graphs that are functions of local neighborhoods. Our results improve all bounds for undirected graphs initially established in~\cite{song18} by a factor of $\sqrt{T} / \log^{3/2} T$ in the additive error.

\section{Preliminaries}
\label{sec:preliminaries}

\subsection{Graphs and Graph Sequences}\label{subsec:defs}
We consider undirected graphs $G = (V,E)$, which change dynamically.
Graphs may be edge-weighted, in which case $G = (V,E,w)$, where $w: E \to \nats$.
The evolution of a graph is described by a \emph{graph sequence} $\graphseq = (G_1,G_2,\dots)$,
where $G_t = (V_t,E_t)$ is derived from $G_{t-1}$ by applying updates, i.e., inserting or deleting nodes or edges.
We denote by $|\graphseq|$ the length of $\graphseq$, i.e., the number of graphs in the sequence.
At time $t$ we delete a set of nodes $\Vdel{t}$ along with the corresponding edges $\Edel{t}$ and
insert a set of nodes $\Vins{t}$ and edges $\Eins{t}$.
More formally, $V_t = (V_{t-1} \setminus \Vdel{t}) \cup \Vins{t}$ and $E_t = (E_{t-1} \setminus \Edel{t}) \cup \Eins{t}$, with initial node and edge sets $V_0$, $E_0$, which may be non-empty.
If a node $v$ is deleted, then all incident edges are deleted at the same time,
i.e., if $v \in \Vdel{t}$, then $(u,v) \in \Edel{t}$ for all $(u,v) \in E_{t-1}$.
Both endpoints of an edge inserted at time $t$ need to be in the graph at time $t$, i.e.,
$\Eins{t} \subseteq ((V_{t-1} \setminus \Vdel{t}) \cup \Vins{t}) \times ((V_{t-1} \setminus \Vdel{t}) \cup \Vins{t})$.
The tuple $(\Vins{t}, \Vdel{t}, \Eins{t}, \Edel{t})$ is the \emph{update} at time $t$. For any graph $G$ and any update $u$, let $\apply{G}{u}$ be the graph that results from applying $u$ on $G$.

A graph sequence is \emph{incremental} if $\Edel{t} = \Vdel{t} = \emptyset$ at all time steps $t$.
A graph sequence is \emph{decremental} if $\Eins{t} = \Vins{t} = \emptyset$ at all time steps $t$.
Incremental and decremental graph sequences are called \emph{partially dynamic}.
Graph sequences that are neither incremental nor decremental are \emph{fully dynamic}.

Our goal is to continually release the value of a \emph{graph function} $f$ which takes a graph as input and outputs a real number.
In other words, given a graph sequence $\graphseq = (G_1,G_2,\dots)$ we want to compute the sequence $f(\graphseq) = (f(G_1),f(G_2),\dots)$.
We write $f(t)$ for $f(G_t)$.
Our algorithms will compute an update to the value of $f$ at each time step, i.e., we compute $\Delta f(t) = f(t) - f(t-1)$. We call the sequence $\Delta f$ the \emph{difference sequence of $f$}.

Given a graph function $g$ the \emph{continuous global sensitivity} $\GS(g)$ of $g$ is defined as the maximum value of $||g(S) - g(S')||_1$ over all adjacent graph sequences $S$, $S'$.
We will define adjacency of graph sequences below.
In our case, we are often interested in the continuous global sensitivity of the difference sequence of a graph function $f$, which is given by the maximum value of $\sum_{t=1}^T |\Delta f_{\graphseq}(t) - \Delta f_{\graphseq'}(t)|$, where $\Delta f_{\graphseq}$ and $\Delta f_{\graphseq'}$ are the difference sequences of $f$ corresponding to adjacent graph sequences $\graphseq$ and $\graphseq'$.

Two graphs are \emph{edge-adjacent} if they differ in one edge.
We also define global sensitivity of functions applied to a single graph. Let $g$ be a graph function.
Its \emph{static global sensitivity} $\GSstatic(g)$ is the maximum value of $|g(G) - g(G')|$ over all edge-adjacent graphs $G$, $G'$.

\subsection{Differential Privacy}
The range of an algorithm $\algo$, $\range(\algo)$, is the set of all possible output values of $\algo$.
We denote the Laplace distribution with mean $\mu$ and scale $b$ by $\Lap(\mu, b)$. If $\mu = 0$, we write $\Lap(b)$.
\begin{definition}[$\eps$-differential privacy]
        \label{def:diff-priv}
        A randomized algorithm $\algo$ is \emph{$\eps$-differentially private} if for any two \emph{adjacent} databases
        $B$, $B'$
        and any $S \subseteq \range(\algo)$ we have
                $\Pr[\algo(B) \in S] \leq e^\eps\cdot\Pr[\algo(B') \in S]$.
        The parameter $\eps$ is called the \emph{privacy loss} of $\algo$.
\end{definition}

To apply \cref{def:diff-priv} to graph sequences we now define adjacency for graph sequences.
First, we define edge-adjacency, which is useful if the data to be protected is associated with the edges in the graph sequence.
Then, we define node-adjacency, which provides stronger privacy guarantees.

\begin{definition}[Edge-adjacency]
        \label{def:edge-adjacency}
        Let $\graphseq$, $\graphseq'$ be graph sequences as defined above with associated sequences of updates
        $(\Vdel{t})$, $(\Vins{t})$, $(\Edel{t})$, $(\Eins{t})$ and
        $(\Vdel{t}')$, $(\Vins{t}')$, $(\Edel{t}')$, $(\Eins{t}')$.
        Let $\Vdel{t} = \Vdel{t}'$ and $\Vins{t} = \Vins{t}'$ for all $t$.
        Let the initial node and edge sets for $\graphseq$ and $\graphseq'$ be $V_0 = V_0'$ and $E_0 = E_0'$.
        Assume w.l.o.g.\ that $\Edel{t}' \subseteq \Edel{t}$ and $\Eins{t}' \subseteq \Eins{t}$ for all $t$.
        The graph sequences $\graphseq$ and $\graphseq'$ are \emph{adjacent on $e^*$} if $|\graphseq| = |\graphseq'|$, there exists an edge $e^*$
        and one of the following statements holds:
        \begin{enumerate}
                \item $\Edel{t} = \Edel{t}'\,\forall\,t$ and 
                      $\exists t^*$ such that $\Eins{t^*} \setminus \Eins{t^*}' = \{e^*\}$ and
                      $\Eins{t} = \Eins{t}'\,\forall\,t\neq t^*$;
                \item $\Eins{t} = \Eins{t}'\,\forall\,t$ and 
                      $\exists t^*$ such that $\Edel{t^*} \setminus \Edel{t^*}' = \{e^*\}$ and
                      $\Edel{t} = \Edel{t}'\,\forall\,t\neq t^*$;
\end{enumerate}
\end{definition}

\emph{Remark.}  If $\graphseq$ and $\graphseq'$ are edge-adjacent, then for any index $i$ the graphs at index $i$ in the two sequences are edge-adjacent.

Two edge-adjacent graph sequences differ in either the insertion or the deletion of a single edge $e^*$.
There are several special cases that fit into this definition.
For example, we may have $\Vdel{t} = \Vdel{t}' = \Vins{t} = \Vins{t'} = \emptyset$, so only edge updates would be allowed.
Similarly, we can use the definition in the incremental setting by assuming $\Vdel{t} = \Vdel{t}' = \Edel{t} = \Edel{t}' = \emptyset$.

The definition of node-adjacency is similar to that of edge-adjacency, but poses additional constraints on the edge update sets.

\begin{definition}[Node-adjacency]
        \label{def:node-adjacency}
        Let $\graphseq$, $\graphseq'$ be graph sequences as defined above with associated sequences of updates
        $(\Vdel{t})$, $(\Vins{t})$, $(\Edel{t})$, $(\Eins{t})$ and
        $(\Vdel{t}')$, $(\Vins{t}')$, $(\Edel{t}')$, $(\Eins{t}')$.
        Assume w.l.o.g.\ that $\Vdel{t}' \subseteq \Vdel{t}$ and $\Vins{t}' \subseteq \Vins{t}$ for all $t$.
        The graph sequences $\graphseq$ and $\graphseq'$ are \emph{adjacent on $v^*$} if $|\graphseq| = |\graphseq'|$, there exists a node $v^*$
        and one of the following statements holds:
        \begin{enumerate}
                \item \label{item:node-adjacency-1}
                      $\Vdel{t} = \Vdel{t}'\,\forall\,t$ and 
                      $\exists t^*$ such that $\Vins{t} \setminus \Vins{t}' = \{v^*\}$ and
                      $\Vins{t} = \Vins{t}'\,\forall\,t\neq t^*$;
                \item $\Vins{t} = \Vins{t}'\,\forall\,t$ and 
                      $\exists t^*$ such that $\Vdel{t} \setminus \Vdel{t}' = \{v^*\}$ and
                      $\Vdel{t} = \Vdel{t}'\,\forall\,t\neq t^*$;
\end{enumerate}
        Additionally, all edges in $\Eins{t}$ and $\Edel{t}$ are incident to at least one node in $\Vins{t}$ and $\Vdel{t}$, respectively.
        Lastly, we require that $\Eins{t}'$ ($\Edel{t}'$) is the maximal subset of $\Eins{t}$ ($\Edel{t}$) that does not contain edges incident to $v^*$.
\end{definition}

We define the following notions of differential privacy based on these definitions of adjacency.
\begin{definition}
        An algorithm is \emph{$\eps$-edge-differentially private} (on event-level) if it is $\eps$-differentially private when considering edge-adjacency. An algorithm is \emph{$\eps$-node-differentially private} (on event-level) if it is $\eps$-differentially private when considering node-adjacency.
\end{definition}

When explicitly stated, we consider a stronger version of $\eps$-differential privacy, which provides adjacency on user-level. While adjacency on event-level only allows two graph sequences to differ in a single update, user-level adjacency allows any number of updates to differ as long as they affect the same edge (for edge-adjacency) or node (for node-adjacency), respectively.
\begin{definition}
        \label{def:user-edge-adj}
        Let $\graphseq = (G_1, \ldots), \graphseq' = (G'_1, \ldots)$ be graph sequences. The two sequences are edge-adjacent \emph{on user-level} if there exists an edge $e^*$ and a sequence of graph sequences $\mathcal{S} = (\graphseq_1, \ldots, \graphseq_\ell)$ so that $\graphseq_1 = \graphseq$,$\graphseq_\ell = \graphseq'$ and, for any $i \in [\ell-1]$, $\graphseq_i$ and $\graphseq_{i+1}$ are edge-adjacent on $e^*$. An algorithm is $\eps$-edge-differentially private on user-level if it is $\eps$-differentially private when considering edge-adjacency on user-level.
\end{definition}

\subsection{Counting Mechanisms}
Some of our algorithms for releasing differentially private estimates of functions on graph sequences rely on algorithms for counting in streams.

A \emph{stream} $\sigma = \sigma(1) \sigma(2) \cdots$ is a string of \emph{items} $\sigma(i) \in \{L_1,\dots,L_2\} \subseteq \integers$,
where the $i$-th item is associated with the $i$-th \emph{time step}.
A \emph{binary stream} has $L_1 = 0$ and $L_2 = 1$.
We denote the length of a stream, i.e., the number of time steps in the stream, by $|\sigma|$.
Stream $\sigma$ and $\sigma'$ are adjacent if $|\sigma| = |\sigma'|$ and if there exists one and only one $t^*$ such that
$\sigma(t^*) \neq \sigma'(t^*)$ and $\sigma(t) = \sigma'(t)$ for all $t \neq t^*$.

A \emph{counting mechanism} $\algo(\sigma)$ takes a stream $\sigma$ and outputs a real number for every time step.
For all time steps $t$, $\algo$'s output at time $t$ is independent of all $\sigma(i)$ for $i > t$.
At each time $t$ a counting mechanism should estimate the count $c(t) = \sum_{i=1}^t \sigma(i)$.

Following Chan \etal\ \cite{chan11} we describe our mechanisms in terms of \emph{p-sums}, which are partial sums of the stream over a time interval.
For a p-sum $p$ we denote the beginning and end of the time interval by $\psstart(p)$ and $\psend(p)$, respectively.
With this notation the value of $p$ is $\sum_{t=\psstart(p)}^{\psend(p)} \sigma(t)$.
To preserve privacy we add noise to p-sums and obtain \emph{noisy p-sums}: given a p-sum $p$, a noisy p-sum is $\hat{p} = p+\gamma$,
where $\gamma$ is drawn from a Laplace-distribution.

\subsection{Sparse Vector Technique}

The \emph{sparse vector technique} (SVT) was introduced by Dwork et~al. \cite{DwoCom09} and was subsequently improved \cite{HarMul10,RotInt10}. SVT can be used to save privacy budget whenever a sequence of threshold queries $f_1, \ldots, f_T$ is evaluated on a database, but only $c \ll T$ queries are expected to exceed the threshold. Here, a threshold query asks whether a function $f_i$ evaluates to a value above some threshold $t_i$ on the input database. Using SVT, only queries that are answered positively reduce the privacy budget. We use the following variant of SVT, which is due to Lyu et~al.

\begin{lemma}[\cite{LyuUnd16}]
        \label{thm:svt}
        Let $\mathcal{D}$ be a database, $\epsilon, \rho, c > 0$ and let $(f_1, t_1), \ldots$ be a sequence of mappings $f_i$ from input databases to $\reals$ and thresholds $t_i \in \reals$, which may be generated adaptively one after another so that $\rho \geq \max_i {\GSstatic(f_i)}$. \Cref{alg:svt} is $\epsilon$-private.
\end{lemma}

\SetKwFunction{FnInitSvt}{InitializeSvt}
\SetKwFunction{FnSvt}{ProcessSvtQuery}
\begin{algorithm}
        \Fn{\FnInitSvt{$\mathcal{D}, \rho, \epsilon, c$}}{
                $\epsilon_1 \gets \epsilon / 2, \zeta \gets \Lap(\rho / \epsilon_1)$, $\epsilon_2 \gets \epsilon - \epsilon_1, count \gets 0$ \;
        }
	\Fn{\FnSvt{$f_i, t_i$}}{
                $\nu_i \gets \Lap(2 c \rho / \epsilon_2)$ \;
                \If{$count \geq c$}{
                        \Return{abort} \;
                }
                \If{$f_i(\mathcal D) + \nu_i \geq t_i + \zeta$}{
                        $count \gets count + 1$, \Return $\top$ \;       
                }
                \Else{
                        \Return $\bot$ \;
                }
	}
	\caption{\label{alg:svt} SVT algorithm \cite{LyuUnd16}}
\end{algorithm}

\section{Mechanisms Based on Continuous Global Sensitivity}
\label{sec:global-sens-mech}
Some of our mechanisms for privately estimating graph functions are based on mechanisms for counting in streams.
In both settings, we compute the sum of a sequence of numbers and we will show that the mechanisms for counting can be transferred to the graph setting.
However, there are differences in the analysis.
In counting, the input streams differ at only one time step. This allows us to bound the difference in the true value between adjacent inputs and leads to low error.
In the graph setting, the sequence of numbers can vary at many time steps. Here however, we use properties of the counting mechanisms to show that the total difference for this sequence can still be bounded, which results in the same error as in the counting setting.

We first generalize the counting mechanisms by Chan \etal\ \cite{chan11} to streams of integers with bounded absolute value,
and then transfer them to estimating graph functions.

\subsection{Non-Binary Counting}
We generalize the counting mechanisms of Chan \etal\ \cite{chan11} to streams of numbers in $\{-L,\dots,L\}$, for some constant $L$.
We view these algorithms as releasing noisy p-sums from which the count can be estimated.
The generic algorithm is outlined in \cref{alg:generic-counting}
\opt{full,conf}{on page~\pageref{alg:generic-counting}}\opt{confpre}{in \cref{sec:fig-appendix}}.

The algorithm releases a vector of noisy p-sums over $T$ time steps, such that at every time step the noisy p-sums needed to estimate the count up to this time are available.
Each of the noisy p-sums is computed exactly once.
See the proof of \cref{cor:graph-bin-mech} for an example on how to use p-sums.

In order to achieve the desired privacy loss the mechanisms need to meet the following requirements.
Let $\algo$ be a counting mechanism. We define $\range(\algo) = \reals^k$, where $k$ is the total number of p-sums used by $\algo$
and every item of the vector output by $\algo$ is a p-sum.
We assume that the time intervals represented by the p-sums in the output of $\algo$ are deterministic and only depend on the length $T$ of the input stream. 
For example, consider any two streams $\sigma$ and $\sigma'$ of length $T$.
The $\ell$-th element of $\algo(\sigma)$ and $\algo(\sigma')$ will be p-sums of the same time interval $[\psstart(\ell),\psend(\ell)]$.
We further assume that the p-sums are computed independently from each other in the following way: $\algo$ computes the true p-sum and then adds noise from $\Lap(z\cdot\eps^{-1})$, where $z$ is a sensitivity parameter.\newcommand{\CountingAlgorithm}[1][htbp]{
\begin{algorithm}[#1]
        \caption{Generic counting mechanism}
        \label{alg:generic-counting}

        \textbf{Input:} Privacy parameter $\eps$, stream $\sigma$, $|\sigma| = T$, of items from the range $\{L_1,\dots,L_2\}$, where $L_1 < L_2$ are publicly known integers \\
        \textbf{Output:} Vector of noisy p-sums $a \in \reals^k$, released over $T$ time steps \\
        \textbf{Initialization:} Determine which p-sums to compute based on $T$ \\
        At each time step $t \in \{1,\dots,T\}$: \;
        \hspace{1cm} Compute new p-sums $p_i,\dots,p_j$ for $t$ \;
        \hspace{1cm} \textbf{For} $\ell = i,\dots,j$: \;
        \hspace{2cm}       $\hat{p}_\ell = p_\ell + \gamma_\ell, \quad \gamma_\ell \sim \Lap((L_2-L_1)\eps^{-1})$ \;
        \hspace{1cm} Release new noisy p-sums $\hat{p}_i,\dots,\hat{p}_j$ \;
\end{algorithm}
}\opt{full}{\CountingAlgorithm}\opt{conf}{\CountingAlgorithm}\newcommand{\errorcorollary}{
To analyze the error of the algorithms we need the following Corollary from Chan \etal\ \cite{chan11}:
\begin{corollary}[Corollary 2.9 from~\cite{chan11}]
        \label{cor:measure-conc}
        Suppose $\gamma_i$ are independent random variables, where each $\gamma_i$ has Laplace distribution $\Lap(b_i)$.
        Let $Y := \sum_i\gamma_i$ and $b_M = \max_i b_i$. Let $\nu \geq \sqrt{\sum_i b_i^2}$.
        Suppose $0 < \delta < 1$ and $\nu > \max\left\{ \sqrt{\sum_i b_i^2}, b_M\sqrt{\ln \frac{2}{\delta}} \right\}$.
        Then, $\Pr\left[|Y| > \nu\sqrt{8\ln\frac{2}{\delta}}\right] \leq \delta$.

        To simplify our presentation and improve readability, we choose $\nu := \sqrt{\sum_i b_i^2}\cdot\sqrt{\ln\frac{2}{\delta}}$
        and use the following slightly weaker result: with probability at least $1-\delta$, the quantity $|Y|$ is at most $O(\sqrt{\sum_i b_i^2} \log \frac{1}{\delta})$.
\end{corollary}
}\opt{full}{\errorcorollary}The next lemma is stated informally in~\cite{chan11}.
\begin{restatable}[Observation 1 from~\cite{chan11}]{lemma}{obsonelemma}
        \label{lem:observation1}
        Let $\algo$ be the counting mechanism given in \cref{alg:generic-counting} 
        that releases $k$ noisy p-sums,
        such that the count at any time step can be computed as the sum of at most $y$ p-sums and
        every item is part of at most $x$ p-sums.
        Assume that the algorithm receives a binary input stream, i.e.~with $L_1=0$ and $L_2=1$.
        Then, $\algo$ is $(x\cdot\eps)$-differentially private,
        and the error is $O(\eps^{-1}\sqrt{y}\log\frac{1}{\delta})$ with probability at least $1-\delta$ at each time step.
\end{restatable}
\newcommand{\obsoneproof}{
\begin{proof}
        The error follows from \cref{cor:measure-conc} since the noise at each time step is the sum of at most $y$ independent random variables with distribution $\Lap(1/\eps)$.

        We show that the vector of all noisy p-sums output by $\algo$ preserves $(x\cdot\eps)$-differential privacy.
        Then, since the count is estimated by summing over $(x\cdot\eps)$-differentially private noisy p-sums the count itself is $(x\cdot\eps)$-differentially private.
        As described above, $\range(\algo) = \reals^k$, where $k$ is the number of p-sums that the algorithm computes in total.

        Let $\sigma$, $\sigma'$ be two adjacent streams.
        We consider the joint distribution of the random variables $Z_1,\dots,Z_k$
        that correspond to the noisy p-sums output by $\algo$.
        Let $p_i(z_i)$, $p_i'(z_i)$ for $z_i \in \range(Z_i)$ be the probability density function of $Z_i$ when the input to $\algo$ is $\sigma$ and $\sigma'$, respectively.
        Note that each of these densities is the probability density function of a Laplace distribution.
        Let $p(z_1,\dots,z_k)$, $p'(z_1,\dots,z_k)$ be the joint probability density when the input to $\algo$ is $\sigma$ and $\sigma'$, respectively.
        Note that the random variables $Z_i$ are continuous, since the noise added to the p-sums is drawn from a continuous distribution over $\reals$.
        Since the $Z_i$ are independent, we have
        \begin{equation*}
                p(z_1,\dots,z_k)  = \prod_{i=1}^k p_i(z_i)
        \end{equation*}
        and
        \begin{equation*}
                p'(z_1,\dots,z_k) = \prod_{i=1}^k p_i'(z_i).
        \end{equation*}

        We now show that these joint densities differ by at most an $\exp(x\cdot\eps)$-factor for all possible outputs of $\algo$.
        Let $s = (s_1,\dots,s_k)\TT \in \range(\algo)$.
        Note that $\algo$ outputs a vector of noisy p-sums from which the count can be computed, so each item in $s$ is a noisy p-sum.
        Let $c = (c_1,\dots,c_k)\TT$ be the noiseless p-sums of stream $\sigma$ corresponding to the noisy p-sums computed by $\algo$.
        Let $I$ be the set of indices to (noisy) p-sums involving time $t'$.
        By assumption, $|I| \leq x$, since any item participates in at most $x$ (noisy) p-sums.
        \begin{align}
                \frac{p(s)}{p'(s)} &= \prod_{i=1}^k \frac{p_i(s_i)}{p_i'(s_i)} 
                                   = \prod_{i \in I} \frac{p_i(s_i)}{p_i'(s_i)} \notag\\
                                   &= \prod_{i \in I} \frac{\exp(-\eps |c_i - s_i|)}{\exp(-\eps |c_i \pm 1 - s_i|)}  \notag\\
                                   &= \prod_{i \in I} \exp(\eps (|c_i \pm 1 - s_i| - |c_i - s_i|))  \notag\\
                                   &\leq \prod_{i \in I} \exp(\eps |c_i \pm 1 - s_i - c_i + s_i|)  \notag\\
                                   &= \exp(\eps)^{|I|} \leq \exp(x\cdot\eps). \label{eq:obs1-density-proof}
        \end{align}
        The first equality follows from the assumption that the noisy p-sums are independent, as discussed above.
        Since all noiseless p-sums that do not involve the item at $t'$ are equal,
        the densities $p_j(s_j)$ and $p_j'(s_j)$ cancel out for $j \notin I$,
        so we can take the product over just the noisy p-sums corresponding to indices in $I$.
        Plugging in the density function of the Laplace distribution and applying the reverse triangle inequality yields the first inequality.
        The second inequality follows from the assumption $|I| \leq x$.

        Now, for any $S \subseteq \range(\algo)$ we have
        \begin{equation*}
                \Pr[\algo(\sigma) \in S] = \int_{S} p(s) ds \leq \int_{S} e^{x\cdot\eps} p'(s)ds = e^{x\cdot\eps} \Pr[A(\sigma') \in S],
        \end{equation*}
        where the inequality follows from \eqref{eq:obs1-density-proof}.
        Thus, $\algo$ is indeed $(x\cdot\eps)$-differentially private.
\end{proof}
} \opt{full}{\obsoneproof}

To extend the counting mechanisms by Chan \etal{} to non-binary streams of values in $\{-L,\dots,L\}$, we only need to account for the increased sensitivity in the scale of the Laplace distribution. By \cref{lem:extended-obs1}, we can use the mechanisms of Chan \etal~\cite{chan11} to compute the sum of a stream of numbers in $\{-L,\dots,L\}$, but gain a factor $2L$ in the error. \opt{confpre}{The proofs appear in \cref{sec:counting-proofs}.}

\begin{restatable}[Extension of \cref{lem:observation1}]{lemma}{obsoneextlemma}
        \label{lem:extended-obs1}
        Let $\algo$ be the counting mechanism given in \cref{alg:generic-counting} 
        that releases $k$ noisy p-sums,
        such that the count at any time step can be computed as the sum of at most $y$ p-sums and
        every item is part of at most $x$ p-sums.
        Assume that the algorithm receives an integer input stream with the guarantee that all input values lie between $-L$ and $L$, where $L$ is a publicly known integer, i.e.~, with $L_1 = -L$ and $L_2 = L$.
        If the input stream complies with the guarantee, then $\algo$ is $(x\cdot\eps)$-differentially private,
        and the error is $O(L\eps^{-1}\sqrt{y}\log\frac{1}{\delta})$ with probability at least $1-\delta$ at each time step.
\end{restatable}
\newcommand{\obsoneextproof}{
\begin{proof}
        The error follows from \cref{cor:measure-conc}.

        We proceed as in the proof of \cref{lem:observation1}.
        Let $\sigma$, $\sigma'$ be two adjacent streams that differ only at time $t'$,
        i.e., $\sigma(t) = \sigma'(t)$ for all $t \neq t'$ and $\sigma(t') \neq \sigma'(t')$.
        We denote by $\delta' = \sigma'(t') - \sigma(t')$ the difference in the items at time $t'$.
        Note that $|\delta'| \leq 2L$, since $\sigma(t),\,\sigma'(t) \in \{-L,\dots,L\}$ for all $t$.
        Let $c = (c_1,\dots,c_k)\TT$ be the noiseless p-sums of stream $\sigma$ corresponding to the noisy p-sums computed by $\algo$.
        Let $I$ be the set of indices to (noisy) p-sums involving time $t'$.
        By assumption, $|I| \leq x$, since any item participates in at most $x$ (noisy) p-sums.

        As in the proof of \cref{lem:observation1} we consider the joint probability density of the output of $\algo$ when the input is $\sigma$ and $\sigma'$.
        Note that the noisy p-sums output by $\algo$ are independent continuous random variables with range $\reals$.
        We have
        \begin{equation*}
                p(z) = \prod_{i=1}^k p_i(z_i)
        \end{equation*}
        and
        \begin{equation*}
                p'(z) = \prod_{i=1}^k p_i'(z_i),
        \end{equation*}
        where the $p_i$ and $p_i'$ are the probability density function of the individual noisy p-sums with input stream $\sigma$ and $\sigma'$, respectively.

        Let $s = (s_1,\dots,s_k)\TT \in \range(\algo)$.
        \begin{align*}
                \frac{p(s)}{p'(s)}
                                   &= \prod_{i=1}^k \frac{p_i(s_i)}{p_i'(s_i)} \\
                                   &= \prod_{i \in I} \frac{p_i(s_i)}{p_i'(s_i)} \\
                                   &= \prod_{i \in I} \frac{\exp\left(-\frac{\eps |c_i - s_i|}{2L}\right)}
                                                           {\exp\left(-\frac{\eps |c_i + \delta' - s_i|}{2L}\right)} \\
                                   &= \prod_{i \in I} \exp\left(\frac{\eps (|c_i + \delta' - s_i| - |c_i - s_i|)}{2L}\right) \\
                                   &\leq \prod_{i \in I} \exp\left(\frac{\eps |c_i + \delta' - s_i - c_i + s_i|}{2L}\right) \\
                                   &= \exp\left(\frac{\eps |\delta'|}{2L}\right)^{|I|} \leq \exp(x\cdot\eps).
        \end{align*}
        Thus, $\algo$ is $x\cdot\eps$-differentially private, since
        \begin{equation*}
                \Pr[\algo(\sigma) \in S] = \int_{S} p(s) ds \leq \int_{S} e^{x\cdot\eps} p'(s)ds = e^{x\cdot\eps} \Pr[A(\sigma') \in S]
        \end{equation*}
        for any $S \subseteq \range(\algo)$.
\end{proof}
} \opt{full}{\obsoneextproof}

\subsection{Graph Functions via Counting Mechanisms}
\begin{algorithm}
        \caption{Generic graph sequence mechanism}
        \label{alg:generic-graph}

        \KwIn{privacy loss $\eps$, publicly known continual global sensitivity $\Gamma$, graph sequence $\graphseq = (G_1,\dots,G_T)$}
        \KwOut{vector of noisy p-sums $a \in \reals^k$, released over $T$ time steps}
        \textbf{Initialization:} Determine which p-sums to compute based on $T$ \;
        At each time step $t \in \{1,\dots,T\}$: \;
        \hspace{1cm} Compute $f(t)$ and $\Delta f(t) = f(t) - f(t-1)$, $f(0) := 0$ \;
        \hspace{1cm} Compute new p-sums $p_i,\dots,p_j$ for the sequence $\Delta f$ \;
        \hspace{1cm} \textbf{For} $\ell = i,\dots,j$: \;
        \hspace{2cm}       $\hat{p}_\ell = p_\ell + \gamma_\ell, \quad \gamma_\ell \sim \Lap(\Gamma\eps^{-1})$ \;
        \hspace{1cm} Release new noisy p-sums $\hat{p}_i,\dots,\hat{p}_j$
\end{algorithm}
We adapt the counting mechanisms to continually release graph functions by following the approach by Song \etal\ \cite{song18}.
\cref{alg:generic-graph} outlines the generic algorithm.
It is similar to \cref{alg:generic-counting}, with the difference that the stream of numbers to be summed is computed from a graph sequence $\graphseq$.
The algorithm is independent of the notion of adjacency of graph sequences, however the additive error is linear in the continuous global sensitivity of the difference sequence $\Delta f$.

In counting binary streams we considered adjacent inputs that differ at exactly one time step.
In the graph setting however, the stream of numbers that we sum, i.e., the difference sequence $\Delta f$, can differ in multiple time steps between two adjacent graph sequences.
We illustrate this with a simple example.
Let $f$ be the function that counts the number of edges in a graph and consider two node-adjacent incremental graph sequences $\graphseq$, $\graphseq'$.
$\graphseq$ contains an additional node $v^*$, that is not present in $\graphseq'$.
Whenever a neighbor is added to $v^*$ in $\graphseq$, the number of edges in $\graphseq$ increases by more than the number of edges in $\graphseq'$.
Thus, every time a neighbor to $v^*$ is inserted, the difference sequence of $f$ will differ between $\graphseq$ and $\graphseq'$.

To generalize this, consider two adjacent graph sequences $\graphseq$, $\graphseq'$ that differ in an update at time $t'$.
Let $\Delta f_\graphseq$ and $\Delta f_{\graphseq'}$ be the difference sequences used to compute the graph function $f$ on $\graphseq$ and $\graphseq'$.
As discussed above we may have $\Delta f_{\graphseq}(t) \neq \Delta f_{\graphseq'}(t)$ for all $t \geq t'$.
Thus, more than $x$ p-sums can be different, which complicates the proof of the privacy loss.
However, we observe that the set $P$ of p-sums with differing values can be partitioned into $x$ sets $P_1,\dots,P_x$, where the p-sums in each $P_i$ cover disjoint time intervals.
By using a bound on the continuous global sensitivity of the difference sequence $\Delta f$, this will lead to a privacy loss of $x\cdot\eps$.
\opt{full}{We now prove that the results on privacy loss and error transfer from the counting mechanisms to the graph mechanisms.}
\opt{conf}{The following lemma formalizes our result.}

\begin{restatable}{lemma}{graphobsonelemma}
        \label{lem:graph-obs1}
        Let $f$ be a function whose difference sequence has continuous global sensitivity $\Gamma$ and $\Gamma$ is publicly known.
        Let $0 < \delta < 1$ and $\eps > 0$.
        Let $\algo$ be a mechanism to estimate $f$ as in \cref{alg:generic-graph} that releases $k$ noisy p-sums
        and satisfies the following conditions:
        \begin{enumerate}
                \item \label{enum:condition-y} at any time step the value of a graph function $f$ can be estimated as the sum of at most $y$ noisy p-sums,
\item \label{enum:condition-independent} $\algo$ adds independent noise from $\Lap(\Gamma/\eps)$ to every p-sum,
                \item \label{enum:condition-disjoint} the set $P$ of p-sums computed by the algorithm can be partitioned
                        into at most $x$ subsets $P_1,\dots,P_x$, such that in each partition $P_x$ all p-sums cover disjoint time intervals.
                        That is, for all $P_i \in \{P_1,\dots,P_x\}$ and all $j,k \in P_i$, $j \neq k$, it holds that
                                (1) $\psstart(j) \neq \psstart(k)$ and
                                (2) $\psstart(j) < \psstart(k) \implies \psend(j) < \psstart(k)$.
        \end{enumerate}
        Then, $\algo$ is $(x\cdot\eps)$-differentially private,
        and the error is $O(\Gamma\eps^{-1}\sqrt{y}\log\frac{1}{\delta})$ with probability at least $1-\delta$ at each time step.
\end{restatable}
\newcommand{\graphobsoneproof}{
\begin{proof}
        As before the error follows from \cref{cor:measure-conc} and condition \ref{enum:condition-y}.

        Let $\graphseq = (G_1,\dots,G_T)$, $\graphseq' = (G_1',\dots,G_T')$ be two adjacent graph sequences and
        let $f(t) = f(G_t)$, $f'(t) = f(G_t')$.
        By definition of adjacent graph sequences (see \cref{def:edge-adjacency,def:node-adjacency}) we have $f(0) = f'(0) = f((V_0,E_0))$.
        At every time step the algorithm computes $\Delta f_{\graphseq}(t) = f(t) - f(t-1)$ and $\Delta f_{\graphseq'}(t) = f'(t) - f'(t-1)$.

        Note that the noisy p-sums computed by $\algo$ are independent continuous random variables with joint distribution
        \begin{equation*}
                p(z) = \prod_{i=1}^k p_i(z_i)
        \end{equation*}
        and
        \begin{equation*}
                p'(z) = \prod_{i=1}^k p_i'(z_i)
        \end{equation*}
        when the input stream is $\graphseq$ and $\graphseq'$, respectively. As in the proofs of \cref{lem:observation1,lem:extended-obs1}, the $p_i$ and $p_i'$ are the marginal probability density functions.
        The random variables are continuous, since the noise added to the p-sums is drawn from a continuous distribution over $\reals$.
        
        Furthermore, let $c = (c_1,\dots,c_k)\TT$ and $c' = (c_1',\dots,c_k')\TT$ be the noiseless p-sums corresponding to the noisy p-sums output by $\algo$ on inputs $\graphseq$ and $\graphseq'$, respectively.
        For each time step $t \in \{1,\dots,T\}$ we define $\delta(t) = \Delta f_{\graphseq'}(t) - \Delta f_{\graphseq}(t)$.
        Note that $\sum_{t=1}^T |\delta(t)| \leq \Gamma$, by definition of the continuous global sensitivity $\Gamma$ of the difference sequence.
        In this proof we use $\psstart(i)$ and $\psend(i)$ to denote the beginning and end of the time interval corresponding to the p-sums with index $i$.
        For each $i \in \{1,\dots,k\}$ we define $\delta_i = c_i' - c_i = \sum_{t=\psstart(i)}^{\psend(i)} \delta(t)$.
        Let $I = \{i \mid \delta_i \neq 0\}$ be the indices of p-sums where the values for the two graph sequences are different.
        Lastly, let $s = (s_1,\dots,s_k)\TT \in \range(\algo)$. 

        \begin{align}
                \frac{p(s)}{p'(s)}
                                   &= \prod_{i=1}^k \frac{p_i(s_i)}{p_i'(s_i)} \notag\\
                                   &= \prod_{i \in I} \frac{p_i(s_i)}{p_i'(s_i)} \notag\\
                                   &= \prod_{i \in I} \frac{\exp\left(-\frac{\eps |c_i - s_i|}{\Gamma}\right)}
                                                           {\exp\left(-\frac{\eps |c_i + \delta_i - s_i|}{\Gamma}\right)} \notag\\
                                   &= \prod_{i \in I} \exp\left(
                                                                   \frac{\eps}{\Gamma} \cdot
                                                                   \left(
                                                                           |c_i + \delta_i - s_i| - |c_i - s_i|
                                                                   \right)
                                                          \right) \notag\\
                                   &\leq \prod_{i \in I} \exp\left(\frac{\eps}{\Gamma}\cdot|\delta_i| \right) \label{eq:graph-loss-1}
        \end{align}
        We use condition \ref{enum:condition-disjoint} and partition $I$ into sets of indices $I_1,\dots,I_x$ such that for all $j \in \{1,\dots,x\}$ the p-sums corresponding to indices in $I_j$ cover disjoint time intervals.
        For each set $I_j$ we then have
        \begin{equation}
                \label{eq:graph-psum-bound}
                \sum_{i\in I_j} |\delta_i| \leq \sum_{i\in I_j}\sum_{t=\psstart(i)}^{\psend(i)} |\delta(t)| \leq \sum_{t=1}^T |\delta(t)| \leq \Gamma.
        \end{equation}
        Combining \eqref{eq:graph-loss-1} and \eqref{eq:graph-psum-bound} yields
        \begin{align*}
                \frac{p(s)}{p'(s)}
                                        &\leq \prod_{i \in I} \exp\left(
                                                                        \frac{\eps}{\Gamma}\cdot |\delta_i|
                                                                  \right) \notag\\
                                        &= \prod_{j = 1}^x \prod_{i \in I_j} \exp\left(
                                                                        \frac{\eps}{\Gamma}\cdot |\delta_i|
                                                                  \right) \notag\\
                                        &= \prod_{j = 1}^x \exp\left(
                                                                     \frac{\eps}{\Gamma}\cdot
                                                                     \sum_{j \in I_j} |\delta_i|
                                                               \right) \notag\\
                                        &\leq \prod_{j = 1}^x \exp\left(\frac{\eps}{\Gamma} \cdot \Gamma\right) = \exp(x\cdot\eps). \label{eq:graph-loss-2}
        \end{align*}
        
        Thus, $\algo$ is $x\cdot\eps$-differentially private, since
        \begin{equation*}
                \Pr[\algo(\sigma) \in S] = \int_{S} p(s) ds \leq \int_{S} e^{x\cdot\eps} p'(s)ds = e^{x\cdot\eps} \Pr[A(\sigma') \in S]
        \end{equation*}
        for any $S \subseteq \range(\algo)$.
        This concludes the proof of \cref{lem:graph-obs1}.
\end{proof}
} \opt{full}{\graphobsoneproof}

\newcommand{\binmechfigure}[1][]{
\begin{figure}[#1]
        \centering
        \begin{tikzpicture}

\begin{scope}
        [bmnode/.style={circle, draw, minimum size = 7mm, inner sep=0pt},
         activenode/.style={circle, draw, ultra thick, pattern = horizontal lines, minimum size = 7mm, inner sep=0pt}]

        \node (1)  at (0.00, 0.00) [bmnode] {};
        \node (2)  at (1.50, 0.00) [activenode] {};
        \node (12) at (0.75, 1.29) [activenode] {};
        \draw (1) to (12);
        \draw (2) to (12);

        \node (3)  at (3.00, 0.00) [bmnode] {};
        \node (4)  at (4.50, 0.00) [activenode] {};
        \node (34) at (3.75, 1.29) [activenode] {};
        \draw (3) to (34);
        \draw (4) to (34);

        \node (14) at (2.25, 2.58) [activenode] {};
        \draw (12) to (14);
        \draw (34) to (14);

        \node (5)  at (6.00, 0.00) [activenode] {};
        \node (6)  at (7.50, 0.00) [bmnode] {};
        \node (56) at (6.75, 1.29) [activenode] {};
        \draw (5) to (56);
        \draw (6) to (56);

        \node (7)  at (9.00, 0.00) [bmnode] {};
        \node (8)  at (10.50, 0.00) [bmnode] {};
        \node (78) at (9.75, 1.29) [bmnode] {};
        \draw (7) to (78);
        \draw (8) to (78);

        \node (58) at (8.25, 2.58) [activenode] {};
        \draw (56) to (58);
        \draw (78) to (58);

        \node (18) at (5.25, 3.87) [activenode] {};
        \draw (14) to (18);
        \draw (58) to (18);

        \node at (-0.75, -1) {$t=$};
        \node at ( 0.00, -1) {1};
        \node[font=\bfseries] at ( 1.50, -1) {2};
        \node at ( 3.00, -1) {3};
        \node[font=\bfseries] at ( 4.50, -1) {4};
        \node[font=\bfseries] at ( 6.00, -1) {5};
        \node at ( 7.50, -1) {6};
        \node at ( 9.00, -1) {7};
        \node at (10.50, -1) {8};

        \begin{pgfonlayer}{bg}
                \node[draw, rectangle, dashed, thick, fill = col1!50, rounded corners=0.25cm,
                        fit = (2) (4) (5), label={[col1]right:$P_0$}] {};
                \node[draw, rectangle, dashed, thick, fill = col2!50, rounded corners=0.25cm,
                        fit = (12) (34) (56), label={[col2]right:$P_1$}] {};
                \node[draw, rectangle, dashed, thick, fill = col3!50, rounded corners=0.25cm,
                        fit = (14) (58), label={[col3]right:$P_2$}] {};
                \node[draw, rectangle, dashed, thick, fill = col4!50, rounded corners=0.25cm,
                        fit = (18), label={[col4]right:$P_3$}] {};
        \end{pgfonlayer}
\end{scope}

\end{tikzpicture}
         \caption{For a stream of length 8 the binary mechanism computes p-sums as shown.
                The leaves of the tree correspond to p-sums over the individual time steps;
                the root corresponds to a p-sum over the full stream.
                Now assume that in \cref{alg:generic-graph} $\Delta f$ differs at time steps 2, 4 and 5 (marked in bold) for adjacent graph sequences.
                The p-sums involving these time steps (shaded circles) can be partitioned into 4 sets $P_0,P_1,P_2,P_3$,
                indicated by dashed rectangles. $P_i$ contains the p-sums of length $2^i$.}
        \label{fig:binary-disjoint}
\end{figure}
} \opt{full}{\binmechfigure[tbp]}
\opt{confpre}{The proof of this \namecref{lem:graph-obs1} and of the following \namecref{cor:graph-bin-mech} appear in \cref{sec:graph-mech-proofs}. }We can compute the p-sums in \cref{alg:generic-graph} as in the binary mechanism \cite{chan11} to release $\eps$-differentially private estimates of graph functions.
\begin{restatable}[Binary mechanism]{corollary}{binmechcorollary}
        \label{cor:graph-bin-mech}
        Let $f$ be a function whose difference sequence has continuous global sensitivity $\Gamma$  and $\Gamma$ is publicly known.
        Let $0 < \delta < 1$ and $\eps > 0$.
        For each $T \in \nats$ there exists an $\eps$-differentially private algorithm to estimate $f$ on a graph sequence
        which has error $O(\Gamma\eps^{-1}\cdot \log^{3/2} T \cdot \log \delta^{-1})$ with probability at least $1-\delta$. 
\end{restatable}
\newcommand{\binmechproof}{
\begin{proof}
        We show that the binary mechanism by Chan \etal\ \cite{chan11} can be employed to construct an algorithm as in \cref{alg:generic-graph}.
        
        The binary mechanism divides the $T$ time steps of the input sequence into $\lfloor T/2^i \rfloor$ intervals of length $2^i$ for each $i = 0,\dots,\log T$.
        For each interval the binary mechanism computes a p-sum.
        The p-sum of length $2^i$ used to compute the output at time step $t$ is indexed as
        \begin{equation*}
                q_i(t) = \sum_{j=0}^{i-1} \flfrac{T}{2^i} + \flfrac{t}{2^i}.
        \end{equation*}
        Let $a=(a_1,\dots,a_k)$ be the noisy p-sums of the difference sequence of a graph function $f$ for any graph sequence of length $T$.
        The value $f(t)$ at time $t$ is computed as the sum
        \begin{equation*}
                \sum_{i=0}^{\log T} \bin_i(t) \cdot a_{q_i(t)},
        \end{equation*}
        where $\bin_i(t)$ is the $i$-th bit in the binary representation of $t$, with $i=0$ for the least significant bit.
        
        We have $y = \log T$, since the sum at any time $t$ is computed from at most $\log T$ noisy p-sums corresponding to the at most $\log T$ bits of the binary representation of $t$.
        Every time step $t$ participates in p-sums of length $2^i$ for all $i = 0,\dots,\log T$. Thus, $x = \log T + 1$.
        The set $P$ of all p-sums can be partitioned into $x$ subsets of disjoint p-sums as illustrated in \cref{fig:binary-disjoint}\opt{conf}{ in \cref{sec:fig-appendix}}:
        for each $i = 0,\dots,\log T$ the set $P_i$ contains the p-sums of length exactly $2^i$.
        Thus, using the binary mechanism to sum the noisy p-sums we obtain an algorithm that satisfies the conditions of \cref{lem:graph-obs1}.

        To obtain an $\eps$-differentially private algorithm we set the privacy parameter in \cref{lem:graph-obs1} to $\eps/x$.
        The error is then $O(\Gamma\eps^{-1}\cdot x \cdot \sqrt{y}\cdot\log\delta^{-1})$ with probability at least $1-\delta$.
        The claim follows from $x,y = O(\log T)$.
\end{proof}
} \opt{full}{\binmechproof}

\subsection{Bounds on Continuous Global Sensitivity}
\label{sec:bounds-global-sensitivity}
Song \etal\ \cite{song18} give bounds on the continuous global sensitivity of the difference sequence for several graph functions in the incremental setting in terms of the maximum degree $D$.
\Cref{tbl:song-sensitivity}\opt{confpre}{ in \cref{sec:gs-bounds-proofs}} summarizes the results on the continuous global sensitivity of difference sequences for a variety of problems in the partially dynamic and fully dynamic setting, both for edge- and node-adjacency.
Note that bounds on the continuous global sensitivity of the difference sequence for incremental graph sequences hold equally for decremental graph sequences as for every incremental graph sequence there exists an equivalent decremental graph sequence which deletes nodes and edges in the reverse order.

In the partially dynamic setting, the continuous global sensitivity based approach works well for graph functions that can be expressed as the sum of local functions on the neighborhood of nodes.
For non-local problems the approach is less successful.
For the weight of a minimum spanning tree the continuous global sensitivity of the difference sequence is independent of the length of the graph sequence. However, for minimum cut and maximum matching this is not the case.
In the fully-dynamic setting the approach seems not to be useful.
Here, we can show that even for estimating the number of edges the continuous global sensitivity of the difference sequence scales linearly with $T$ under node-adjacency.
When considering edge-adjacency we only have low sensitivity for the edge count.
\opt{confpre}{The proofs are in \cref{sec:gs-bounds-proofs}.}

\newcommand{\sensitivitytable}{
\begin{table}
\caption{Global sensitivity of difference sequences for graphs with maximum degree $D$\opt{confpre}{. Proofs are in \cref{sec:gs-bounds-proofs,sec:mst-algo-proofs}.}}
\label{tbl:song-sensitivity}
\centering
\begin{tabulary}{\linewidth}{LRRRR}
        \toprule
        Graph Function $f$    & \multicolumn{4}{c}{Continuous Global Sensitivity of $\Delta f$}           \\
        \midrule
                              & \multicolumn{2}{c}{Partially Dynamic} & \multicolumn{2}{c}{Fully Dynamic} \\
                              & node-adjacency                        & edge-adjacency         & node-adj. & edge-adj. \\
        \midrule
        edge count\afoot        & $D$                                   & $1$                    & $\geq T$  & 2         \\
        high-degree nodes\afoot & $2D + 1$                              & $4$                    & $\geq T$  & $\geq T$ \\
        degree histogram\afoot  & $4D^2 + 2D + 1$                       & $8D$                   & $\geq 2T$ & $\geq 2T$ \\
        triangle count\afoot    & $\binom{D}{2}$                        & $D$                    & $\geq T$  & $\geq T$ \\
        $k$-star count\afoot    & $D \binom{D-1}{k-1} + \binom{D}{k}$   & $2\cdot \bigl(\binom{D}{k} - \binom{D-1}{k}\bigr)$
                                                                                               & $\geq T$  & $\geq T$ \\
        minimum spanning tree & $2DW$                                 & $2W-2$                 & $\geq T$  & $\geq T$ \\
        minimum cut           & $\geq T$                              & $\geq T$               & $\geq T$  & $\geq T$ \\
        maximum matching      & $\geq T$                              & $\geq T$               & $\geq T$  & $\geq T$ \\
        \bottomrule
        \multicolumn{5}{l}{\afoot\footnotesize{Bounds for partially dynamic node-adjacency from Song et al.\ \cite{song18}}}
\end{tabulary}
\end{table}
} \opt{full}{\sensitivitytable}
\opt{conf}{\sensitivitytable}

Using \cref{cor:graph-bin-mech} we obtain $\eps$-differentially private mechanisms with additive error that scales with $\log^{3/2} T$,
compared to the factor $\sqrt{T}$ in \cite{song18}.
Note that we recover their algorithm when using the Simple Mechanism II by Chan \etal\ \cite{chan11} to sum the difference sequence.

\newcommand{\GSproofs}{
For edge-adjacency we obtain the following bounds on the sensitivity of the difference sequence.
\begin{lemma}
        \label{lem:pd-edge-sens}
        When considering edge-adjacency in the incremental setting, the continuous global sensitivity of the difference sequence of the
        \begin{enumerate}
                \item number of high-degree nodes is 4;
                \item degree histogram is $8D$;
                \item edge count is 1;
                \item triangle count is $D$;
                \item $k$-star count is $2\cdot \bigl(\binom{D}{k} - \binom{D-1}{k}\bigr)$.
        \end{enumerate}
\end{lemma} 
\begin{proof}
        In the following, $\graphseq$ and $\graphseq'$ are edge-adjacent incremental graph sequences that differ in the insertion of an edge $e^* = \{u^*, v^*\}$.
        
        \begin{enumerate}
                \item In $\graphseq$, the endpoints of $e^*$ may cross the degree-threshold $\tau$ as soon as $e^*$ is inserted, increasing the number of high-degree nodes by 2.
                In $\graphseq'$, the same nodes may become high-degree nodes at a later update, increasing the number of high-degree nodes by 1 each.
                Thus, the difference sequences differ by at most 4.

                \item The degrees of $u^*$ and $v^*$ are one less in $\graphseq'$ compared to $\graphseq$, once $e^*$ is inserted.
                There are up to $D$ insertions of edges incident to $u^*$, $v^*$.
                For each of these updates, the increase in the degree of $u^*$ affects 2 histogram-bins in $\graphseq$ and $\graphseq'$, each.
                The same holds for the increase in the degree of $v^*$.
                Thus, the continuous global sensitivity of the difference sequence of the degree histogram is $8D$.

                \item The edge-update sets $\Eins{t}$ and $\Eins{t}'$ have exactly the same size at all times, except when $e^*$ is inserted.
                Thus, the global sensitivity of the difference sequence the edge count is the number of edges that are inserted into $\graphseq$, but not into $\graphseq'$, which is 1.

                \item In $\graphseq'$, the $u^*$ and $v^*$ are not part of shared triangles.
                In $\graphseq$ however, they may become part of up to $D$ triangles.
                Thus, the continuous global sensitivity of the difference sequence of the triangle count is $D$.

                \item In $\graphseq$, after inserting $e^*$, $u^*$ and $v^*$ may be the center of up to $\binom{D}{k}$-many $k$-stars.
                Then, they were the center of $\binom{D-1}{k}$-many $k$-stars before the insertion of $e^*$.
                Any $k$-star in $\graphseq'$ is also present in $\graphseq$.
                Thus, the continuous global sensitivity of the difference sequence of the $k$-star count is the number of $k$-stars that the insertion of $e^*$ can contribute, which is $2 \bigl(\binom{D}{k} - \binom{D-1}{k}\bigr)$. \qedhere
        \end{enumerate}
\end{proof}

The following lemmata provide evidence that the difference sequence based approach is not useful for fully dynamic graphs.
\begin{lemma}
        \label{lem:high-degree-fd-sens}
        Let $f_\tau((V,E)) = |\{v \in V \mid \deg(v) \geq \tau\}|$.
        The difference sequence of $f_\tau$ has continuous global sensitivity at least $T$ when considering edge-adjacent or node-adjacent fully dynamic graph sequences.
\end{lemma}
\begin{proof}
        For the edge-adjacency case consider two graph sequences $\graphseq$, $\graphseq'$,
        where the initial graph is a $\tau$-star with two edges $e_1$, $e_2$ removed.
        At time $t=1$ insert $e_1$ into $\graphseq$, but not into $\graphseq'$.
        At odd times, insert $e_2$ into both sequences; at even times delete $e_2$ from both sequences.
        The sequence $\graphseq'$ never has a node of degree at least $\tau$,
        whereas $\graphseq$ has one such node at odd, but not at even time steps.
        Thus, $\sum_{t=1}^T |\Delta f_{\tau,\graphseq}(t) - \Delta f_{\tau,\graphseq'}(t)| = T$.

        For the node-adjacency case we again consider two graph sequences $\tilde\graphseq$, $\tilde\graphseq'$,
        where the initial graph is a $\tau$-star with two nodes $a$, $b$ and their incident edges removed.        
        At time $t=1$ insert $a$ and its incident edge into $\tilde\graphseq$, but not into $\tilde\graphseq'$.
        At odd times, insert $b$ and its incident edge into both sequence; at even times delete $b$ from both sequences.
        As in the edge-adjacency case it follows that
        $\sum_{t=1}^T |\Delta f_{\tau,\tilde\graphseq}(t) - \Delta f_{\tau,\tilde\graphseq'}(t)| = T$.
\end{proof}
\begin{lemma}
        \label{lem:degree-histogram-fd-sens}
        The difference sequence of the degree histogram has continuous global sensitivity at least $2T$ when considering edge-adjacent or node-adjacent fully dynamic graph sequences.
\end{lemma}
\begin{proof}
        We refer to the adjacent graph sequences introduced in the proof of \cref{lem:high-degree-fd-sens}.
        In one of the graph sequences, the bins in the degree histogram corresponding to degrees $\tau-1$ and $\tau$ are alternatingly zero and one. 
        In the other graph sequence the same is true for the bins corresponding to degrees $\tau-2$ and $\tau$.
        Let $h_\graphseq(t)$, $h_{\graphseq'}(t)$ be the degree histograms of the former and latter sequence at time $t$, respectively.
        The corresponding difference sequences are $\Delta h_{\graphseq}(t) = (d_{t,0}, \dots, d_{t,n-1})$ and $\Delta h_{\graphseq'}(t) = (d_{t,0}',\dots,d_{t,n-1}')$,
        where $d_{t,i}$ and $d_{t,i}'$ is the change in the number of nodes of degree $i$ from time $t-1$ to $t$.
        For odd $t$, we have $d_{t,\tau} = 1$, $d_{t,\tau-1} = -1$, $d_{t,\tau-1}' = 1$, $d_{t,\tau-2}' = -1$;
        for even $t$, we have $d_{t,\tau} = -1$, $d_{t,\tau-1} = 1$, $d_{t,\tau-1}' = -1$, $d_{t,\tau-2}' = 1$.
        Note that $d_{t,1} = d_{t,1}'$ for $t \geq 2$ and $d_{1,1} = 2 = d_{1,1'} + 1$.
        All $d_{t,i}$ not explicitly given are 0.
        It follows that $\sum_{t=1}^T ||\Delta h_{\graphseq}(t) - \Delta h_{\graphseq'}(t) ||_1 \geq 2T$, which concludes the proof.
\end{proof}
\begin{lemma}
        \label{lem:subgraph-count-sensitivity}
        Let $H$ be a connected graph with at least two edges (nodes) and $f_H(G)$ be the function counting occurrences of $H$ in $G$.
        When considering edge-adjacent (node-adjacent) fully dynamic graph sequences of length $T$, $\GS(\Delta f_H) \geq T$.
\end{lemma}
\begin{proof}
        To show the edge-adjacency case consider adjacent graph sequences $\graphseq$, $\graphseq'$ of length $T$,
        where the initial graph is a copy of $H$ with two edges, say $\{a,b\}$, $\{c,d\}$, removed.
        In the first time step ($t=1$), insert $\{a,b\}$ into $\graphseq$, but not into $\graphseq'$ and insert $\{c,d\}$ into both sequences.
        Thus, we have $f_H(G_1) = 1$ and $f_H(G_1') = 0$.
        In time steps where $t$ is even, delete $\{c,d\}$ from both sequences; if $t$ is odd, insert $\{c,d\}$ into both sequences.
        At all time steps $t$ we have $f_H(G_t') = 0$ and $|\Delta f_{H,\graphseq}(t)| = 1$,
        which implies $\sum_{t=1}^T |\Delta f_{H,\graphseq}(t) - \Delta f_{H,\graphseq'}(t)| = T$.

        For the node-adjacency case consider again adjacent graph sequences $\tilde\graphseq$, $\tilde\graphseq'$ of length $T$,
        where the initial graph is a copy of $H$ with two nodes, say $a$, $b$, removed.
        At time step $t=1$ insert $a$ along with the incident edges into $\tilde\graphseq$, but not $\tilde\graphseq'$.
        At every odd time step, insert $b$ along with the incident edges into both sequences; at every even time step delete $b$ from both sequences.
        Again, $\tilde\graphseq'$ never contains a copy of $H$, whereas $\tilde\graphseq$ contains a copy of $H$ at every odd time step,
        but no on the even time steps.
        Thus, $\sum_{t=1}^T |\Delta f_{H,\tilde\graphseq}(t) - \Delta f_{H,\tilde\graphseq'}(t)| = T$.
\end{proof}
\begin{corollary}
        \label{cor:subgraph-count-fd-sens}
        The difference sequence of triangle count and $k$-star count has continuous global sensitivity $\geq T$ when considering edge-adjacent or node-adjacent graph sequences.
        The difference sequence of the edge count has continuous global sensitivity $\geq T$ when considering node-adjacent graph sequences.
\end{corollary}
\begin{proof}
        The claim follows from \cref{lem:subgraph-count-sensitivity},
        by observing that triangles and $k$-stars have at least two edges and nodes,
        and noting that for the edge count $H = (\{a,b\},\{\{a,b\}\})$.
\end{proof}
\begin{lemma}
        \label{lem:mst-fd-sens}
        The difference sequence of the weight of a minimum spanning tree has continuous global sensitivity at least $T$ when considering edge-adjacent or node-adjacent fully dynamic graph sequences.
\end{lemma}
\begin{proof}
        Let $A = (V_A, E_A)$ and $B = (V_B, E_B)$ be connected graphs with $|V_A|,|V_B| \geq 3$ and $w(e) = 1$ for all $e \in E_A \cup E_B$.
        We denote nodes in $V_A$ by $a_i$, nodes in $V_B$ by $b_i$ with index $i=0,1,\dots$.
        Let $G_0 = (V_A \cup V_B, E_A \cup E_B \cup \{\{a_0,b_0\}\})$ with $w(\{a_0,b_0\}) = W \geq 3$.
        Let $\graphseq$, $\graphseq'$ be graph sequences of length $T$ with initial graph $G_0$.

        At time $t=1$ we insert the edge $e^* = \{a_1,b_1\}$ with weight $w(e^*) = 1$ into $\graphseq$.
        Then, if $t$ is even, we insert the edges $e' = \{a_2,b_2\}$ with weight $w(e') = W-1$ into $\graphseq$ and $\graphseq'$.
        If $t$ is odd, we delete $e'$ from both sequences.
        
        The sequences $\graphseq$ and $\graphseq'$ are edge-adjacent.
        The weight of the minimum spanning tree in $G_0$ is $w_0 = |V_A| - 1 + |V_B| - 1 + W = |V_A| + |V_B| + W - 2$.
        For all $t \geq 1$, the weight of the minimum spanning tree in $\graphseq$ is
        $w_\graphseq = 1 + |V_A| - 1 + |V_B| - 1 = |V_A| + |V_B| - 1$.
        If $t$ is even and $t \geq 2$, then the weight of the minimum spanning tree in $\graphseq'$ is $w_0 - 1$;
        if $t$ is odd it is $w_0$.
        Thus, we have
        \begin{equation*}
                |\Delta \wmst_\graphseq(t) - \Delta \wmst_{\graphseq'}(t)| \geq 1 
        \end{equation*}
        for all $t = 1,\dots,T$,
        where $\wmst_\graphseq$ and $\wmst_{\graphseq'}$ denote the weight of the minimum spanning tree in $\graphseq$ and $\graphseq'$, respectively.
        Summing over all $t$ proves the claim for edge-adjacent sequences.

        To show that the bound holds for node-adjacent sequences, we replace the edge-insertions by node-insertions as follows:
        instead of $e^*$ we insert a node $v^*$ along with edges $e^*_1 = \{a_1,v^*\}$ and $e^*_2 = \{v^*,b_1\}$,
        where $w(e^*_1) = w(e^*_2) = 1$;
        instead of $e'$ we insert a node $v'$ along with edges $e'_1 = \{a_2,v'\}$ and $e'_2 = \{v',b_1\}$,
        where $w(e'_1) = 1$ and $w(e'_2) = W - 1$;
        where we delete $e'$ in the edge-adjacency case we delete $v'$ and the incident edges.
        We use the same initial graph $G_0$ and obtain node-adjacent graph sequences where the weights of the minimum spanning tree are increased by 1 compared to the weights in the corresponding node-adjacent graph sequences.
        The claim for node-adjacent graph sequences follows as for edge-adjacent graph sequences.
\end{proof}

For maximum cardinality matching and minimum cut the continuous global sensitivity of the difference sequence scales linearly with the length of the input even in the partially dynamic setting.
\begin{lemma}
        Let $M(G)$ be the size of a maximum cardinality matching of $G$.
        The difference sequence of $M$ has continuous global sensitivity at least $T$ when considering edge-adjacent or node-adjacent partially dynamic graph sequences.
\end{lemma}
\begin{proof}
        First, we construct edge-adjacent graph sequences $\graphseq$, $\graphseq'$,
        where $|\Delta M_\graphseq(t) - \Delta M_{\graphseq'}(t)| = 1$ for all $t$.
        Let $V = \{0,\dots,T\}$.
        At time $t = 1$ insert the edge $\{0,1\}$ into $\graphseq$ and insert no edge into $\graphseq'$.
        $\Delta M_\graphseq(1) = 1$, since $G_1$ has a maximum cardinality matching of size 1.
        $\Delta M_{\graphseq'}(1) = 0$, since $G_1'$ has no edges.
        At every time $t > 1$ insert the edge $\{t-1,t\}$ into both $\graphseq$ and $\graphseq'$.
        If $t$ is even, then $\Delta M_{\graphseq}(t) = 0$ and $\Delta M_{\graphseq'}(t) = 1$.
        If $t$ is odd, then $\Delta M_{\graphseq}(t) = 1$ and $\Delta M_{\graphseq'}(t) = 0$.
        Thus, we have $\sum_{t=1}^T |\Delta M_\graphseq(t) - \Delta M_{\graphseq'}(t)| = T$.

        For node-adjacency, both graph sequences start with the node 1.
        In $\graphseq$, we insert node $0$ along with the edge $\{0,1\}$ at time 1.
        Then, at every time step $t > 1$ we insert node $t$ and the edge $\{t, t-1\}$ into both graph sequences.
        The difference sequence is the same as in the edge-adjacency case.
        Thus, the continuous global sensitivity of the difference sequence is also at least $T$.
\end{proof}
\begin{lemma}
        Let $\wcut(G)$ be the weight of a minimum cut of $G$.
        The difference sequence of $\wcut$ has continuous global sensitivity at least $T$ when considering edge-adjacent or node-adjacent partially dynamic graph sequences.
\end{lemma}
\begin{proof}
        We show that for any $T \in \nats$ there exist edge-adjacent graph sequences $\graphseq$, $\graphseq'$, such that
        $\sum_{t=1}^T |\Delta \wcut_\graphseq (t) - \Delta \wcut_{\graphseq'}(t)| \geq T$.
        
        Let $A = (V_A, E_A)$, $B = (V_B, E_B)$, $C = (V_C, E_C)$ be graphs on at least $T+2$ nodes with minimum cut of weight at least $W\cdot (T+1)$.
        We denote nodes in $V_A$ by $a_i$, nodes in $V_B$ by $b_i$ and nodes in $V_C$ by $c_i$, with index $i = 0,1,\dots$.
        Let $G_0 = (V_A \cup V_B \cup V_C, E_A \cup E_B \cup E_C \cup \{e_{ab}, e_{bc}\}$, where
        $e_{ab} = \{a_{T+1},b_{T+1}\}$ and $e_{bc} = \{b_{T+1},c_{T+1}\}$ with $w(e_{ab}) = 1$ and $w(e_{bc}) = 2$.
        Let $\graphseq$, $\graphseq'$ be graph sequences with initial graph $G_0$.
        
        At time $t=1$ we insert the edge $e^* = \{a_0,b_0\}$ with $w(e^*) = W$ into $\graphseq$.
        Then, if $t$ is even, we insert the edge $e_t = \{a_t, b_t\}$ with $w(e_t) = W$ into both $\graphseq$ and $\graphseq'$;
        If $t$ is odd, we insert the edge $e_t = \{b_t,c_t\}$ with $w(e_t) = W$ into both $\graphseq$ and $\graphseq'$.
        
        These graph sequences are edge-adjacent since they differ only in the insertion of $e^*$.
        At any time $t$ the minimum cut in $G_t$ is $(V_A\cup V_B, V_C)$.
        In $\graphseq'$, the minimum cut is $(V_A, V_B\cup V_C)$ for odd $t$ and $(V_A\cup V_B, V_C)$ for even $t$.
        
        In $\graphseq$, the value of the minimum cut increases by 1 at time $t=1$ due to insertion of $e^*$
        and by $W$ for odd $t \geq 3$, since an edge of weight $W$ across the minimum cut is inserted.
        That is,
        \begin{equation}
                \label{eq:g-cut-delta}
                \Delta \wcut_\graphseq(t) = \begin{cases}
                                                1       & \text{if } t = 1, \\
                                                0       & \text{if } t \text{ is even,} \\
                                                W       & \text{if } t \text{ is odd and } t \geq 3. \\
                                            \end{cases}
        \end{equation}

        In $\graphseq'$, the value of the minimum cut does not change at time $t=1$, since no edge is inserted.
        At even times, the cut switches to the edges between $V_B$ and $V_C$ and the weight of the minimum cut increases by $1$.
        At odd times, the cut switches to the edges between $V_A$ and $V_B$, and the weight of the minimum cut increases by $W-1$.
        The difference sequence for $\graphseq'$ is thus,
        \begin{equation}
                \label{eq:gprime-cut-delta}
                \Delta \wcut_{\graphseq'}(t) = \begin{cases}
                                                0       & \text{if } t = 1, \\
                                                1       & \text{if } t \text{ is even,} \\
                                                (W-1)   & \text{if } t \text{ is odd and } t \geq 3. \\
                                            \end{cases}
        \end{equation}
        From \eqref{eq:g-cut-delta} and $\eqref{eq:gprime-cut-delta}$ we get
        $|\Delta \wcut_\graphseq(t) - \Delta \wcut_{\graphseq'}(t)| = 1$ for all $t$.
        Summing over all $t$ we get the proposed bound on the continuous global sensitivity.

        We obtain the same bound for node-adjacent partially dynamic graph sequences by modifying the above construction as follows:
        we start with the same initial graph $G_0$ as above.
        We replace an insertion of the edge $\{a_i,b_i\}$ by an insertion of node $a_i'$
        along with edges $\{a_i',b_i\}$ and edges $\{a_i', a_j\}$ for all $a_j \in V_A$;
        all these edges have weight $W$.
        Similarly, we replace an insertion of the edge $\{b_i,c_i\}$ by an insertion of node $c_i'$
        along with edges $\{c_i',b_i\}$ and edges $\{c_i', c_j\}$ for all $c_j \in V_C$;
        all these edges have weight $W$.
        By counting the nodes $a_i'$ and $c_i'$ to $V_A$ and $V_C$, respectively,
        we obtain the same sequence of minimum cuts as in the edge-adjacency case.
        Thus, we get the same lower bound on global sensitivity for the node-adjacency case.
\end{proof}
}
\opt{full}{\GSproofs} 

The difference sequence approach can be employed to privately estimate the weight of a minimum spanning tree in partially dynamic graph sequences.
If the edge weight is bounded by $W$, then the continuous global sensitivity of the difference sequence is $O(W)$ and $O(DW)$ under edge-adjacency and node-adjacency, respectively.
\opt{full}{In \cref{thm:mst-gs-edge,thm:mst-gs-node} we prove these bounds and use \cref{cor:graph-bin-mech} to derive edge- and node-differentially private algorithms.}
\opt{conf}{Using \cref{cor:graph-bin-mech} we obtain the algorithms outlined in \cref{thm:mst-gs-edge,thm:mst-gs-node}.}

\begin{restatable}{theorem}{mstedgegslemma}
        \label{thm:mst-gs-edge}
        Let $\graphseq$ be an incremental graph sequence such that  each edge weight belongs to the set $\{1, \dots, W\}$, where $W$ is a publicly known parameter.
        There exists an $\eps$-edge-differentially private algorithm that outputs the weight of a minimum spanning tree on $\graphseq$.
        At every time step, the algorithm has error $O(W\eps^{-1}\cdot \log^{3/2} T \cdot \log \delta^{-1})$ with probability at least $1-\delta$, where $T$ is the length of the graph sequence.
\end{restatable}
\newcommand{\mstedgegsproof}{
\begin{proof}
        We will show that the difference sequence for the weight of a minimum spanning tree has global edge-sensitivity $2W-2$ in the incremental setting. The claim then follows by \cref{cor:graph-bin-mech}.

        Let $\graphseq = (G_1,\dots,G_T)$, $\graphseq' = (G_1',\dots,G_T')$ be two adjacent graph sequences with initial graphs $G_0$ and $G_0'$.
        We are considering edge-differential privacy and thus edge-adjacency (\cref{def:edge-adjacency}).
        This means that $\graphseq$ and $\graphseq'$ differ in the insertion of an edge $e^*$.
        Without loss of generality, we assume that $e^*$ is inserted into $\graphseq$.
        For estimating $\GS(\Delta \wmst)$ we can further assume that $e^*$ is inserted into $G_0$ to obtain $G_1$, since $\Delta \wmst_\graphseq(t) = \Delta\wmst_{\graphseq'}(t)$ for all times $t$ before the insertion of $e^*$.
        For all $t = 0,\dots,T$ the edge set of $G_t$ is a superset of the edge set of $G_t'$.
        This implies $\wmst(G_t) \leq \wmst(G_t')$ for all $t$.

        Note that $\wmst(G_0) = \wmst(G_0') = \wmst(G_1')$, which implies $\Delta \wmst{\graphseq'}(1) = 0$.

                As $e^*$ might either not replace any edge in a minimum spanning tree or might replace an edge of weight $W$, it holds that
                $0 \le \Delta \wmst{\graphseq'}(1) \leq W - 1$. Thus, it follows that $0 \le \Delta\wmst_{\graphseq}(1) - \Delta\wmst_{\graphseq'}(1)\leq W-1$.

        It remains to bound $\sum_{t=2}^T |\Delta\wmst_{\graphseq}(t) - \Delta\wmst_{\graphseq'}(t)|$.
        Any edge inserted into $\graphseq$ is also inserted into $\graphseq'$.
        If inserting an edge into $\graphseq$ at time $t$ reduces the weight of the minimum spanning tree of $\graphseq$ at time $t$,
        then the weight of the minimum spanning tree in $\graphseq'$ at time $t$ is reduced by at least the same amount.
        Thus there are two cases: (1) the weight of the minimum spanning tree in $\graphseq$ and $\graphseq'$ change by the same amount;
        (2) the weight of the minimum spanning tree in $\graphseq'$ is reduced more than the weight of a minimum spanning tree in $\graphseq$.
                It follows that for all $t > 0$, $\Delta\wmst_{\graphseq}(t) - \Delta\wmst_{\graphseq'}(t) \ge 0$.
                Thus,
                $\sum_{t=2}^T |\Delta\wmst_{\graphseq}(t) - \Delta\wmst_{\graphseq'}(t)|$
                $= \sum_{t=2}^T \Delta\wmst_{\graphseq}(t) - \Delta\wmst_{\graphseq'}(t)$
								$=\sum_{t=2}^T \Delta\wmst_{\graphseq}(t) - \sum_{t=2}^T \Delta\wmst_{\graphseq'}(t)$.
								But then it follows that
								\begin{equation*}
                \sum_{t=2}^T \Delta\wmst_{\graphseq}(t) - \sum_{t=2}^T \Delta\wmst_{\graphseq'}(t)
                =\wmst(G_T) - \wmst(G_1) - \wmst(G_T') + \wmst(G_1') \le W-1,
								\end{equation*}
                where the last inequality follows from $\wmst(G_1') - \wmst(G_1) \leq W-1$ and $\wmst{G_{T}} \le \wmst{G_{T}'}$. 

        Combining the results for $t=1$ and $t>1$, we get
        \begin{equation*}
                \sum_{t=1}^T |\Delta\wmst_{\graphseq}(t) - \Delta\wmst_{\graphseq'}(t)| \leq 2W - 2,
        \end{equation*}
        which implies $\GS(\Delta\wmst)\leq 2W-2$.

        We now show that this bound on the sensitivity is tight.
        Let $A = (V_A = \{a_1,\dots\},E_A)$, $B=(V_B = \{b_1,\dots\},E_B)$ be arbitrary connected graphs
        with $|A|,|B| > W+1$ and $w(e) = 1$ for all $e \in E_A \cup E_B$.
        Let $G = (V_A \cup V_B, E_A \cup E_B \cup \{a_1,b_1\})$ and $w(\{a_1,b_1\}) = W$.
        Note that the edge $\{a_1,b_1\}$ is guaranteed to be in any minimum spanning tree.

        We construct adjacent graph sequences $\graphseq_a = (G_{a1}, G_{a2},G_{a3})$, $\graphseq_b = (G_{b1}, G_{b2},G_{b3})$
        with $G_{a1} = G_{b1} = G$,
        such that $\sum_{t=1}^{|\graphseq_a|} |\Delta\wmst_{\graphseq_a}(t) - \Delta\wmst_{\graphseq_b}(t)| = 2W-2$.
        
        Starting with $G_{a1} = G$ we insert the edge $\{a_2,b_2\}$ with $w(\{a_2,b_2\}) = 1$ into $\graphseq_a$ to obtain $G_{a2}$.
        In $\graphseq_b$ we insert no edges, i.e., $G_{b2} = G_{b1} = G$.
        Inserting the unit-weight edge into $\graphseq_a$ reduces the weight of a minimum spanning tree by $W-1$.
        In $\graphseq_b$, the weight of a minimum spanning tree does not change.
        Thus, we have $|\Delta\wmst_{\graphseq_a}(2) - \Delta\wmst_{\graphseq_b}(2)| = W-1$.
        
        In the subsequent time step we insert the edge $\{a_3,b_3\}$ with $w(\{a_3,b_3\}) = 1$ into $\graphseq_a$ and $\graphseq_b$,
        which yields the graphs $G_{a3}$ and $G_{b3}$, respectively.
        This does not change the weight of the minimum spanning tree in $\graphseq_a$,
        but reduces the weight of the minimum spanning tree in $\graphseq_b$ by $W-1$,
        since $\{a_3,b_3\}$ replaces $\{a_1,b_1\}$ in the minimum spanning tree.
        Thus, we have $|\Delta\wmst_{\graphseq_a}(3) - \Delta\wmst_{\graphseq_b}(3)| = W-1$.
        
        This concludes the proof.
\end{proof}
} \opt{full}{\mstedgegsproof}

\begin{restatable}{theorem}{mstnodegslemma}
        \label{thm:mst-gs-node}
        Let $\graphseq$ be an incremental graph sequence such that every node is incident to at most $D$ edges and each edge weight belongs to the set $\{1, \dots, W\}$, where $D$ and $W$ are publicly known parameters.
        There exists an $\eps$-node-differentially private algorithm that outputs the weight of a minimum spanning tree on $\graphseq$.
        At every time step, the algorithm has error $O(DW\eps^{-1}\cdot \log^{3/2} T \cdot \log \delta^{-1})$ with probability at least $1-\delta$, where $T$ is the length of the graph sequence.
\end{restatable}
\newcommand{\mstnodegsproof}{
\begin{proof}
        We will show that the difference sequence for the weight of a minimum spanning tree has continuous global sensitivity $2DW$ under node-adjacency in the incremental setting. The claim then follows by \cref{cor:graph-bin-mech}.

        Let $\graphseq = (G_1,\dots,G_T)$, $\graphseq' = (G_1',\dots,G_T')$ be two adjacent graph sequences with initial graphs $G_0$ and $G_0'$.
        We are considering node-differential privacy and thus node-adjacency (\cref{def:node-adjacency}).
        This means that $\graphseq$ and $\graphseq'$ differ in the insertion of a node $v^*$.
        Without loss of generality, we assume that $v^*$ is inserted into $\graphseq$.
        For estimating $\GS(\Delta \wmst)$ we can further assume that $v^*$ is inserted into $G_0$ to obtain $G_1$, since $\Delta \wmst_\graphseq(t) = \Delta\wmst_{\graphseq'}(t)$ for all times $t$ before the insertion of $v^*$.

        The node $v^*$ has degree at most $D$, so we have $|\Delta\wmst_\graphseq(1) - \Delta\wmst_{\graphseq'}(1)| \leq DW$.
        
        We now bound $\sum_{t=2}^T |\Delta\wmst_\graphseq(t) - \Delta\wmst_{\graphseq'}(t)| =: \Gamma_2$.
        Let $E^*$ be the set of edges incident to $v^*$. Note that $|E^*| \leq D$.
        For each edge $e \in E^*$ there are three cases:
        (1) $e$ is not in the minimum spanning tree when it is inserted;
        (2) $e$ replaces a heavier edge;
        (3) $e$ increases the weight of the minimum spanning tree.
        
        In case (1), there always exists a minimum spanning tree that does not contain $e$, since no edges are ever deleted.
                Thus, $e$ does not contribute to the continuous global sensitivity of the difference sequence.
        To analyze case (2), let $e_H$ be the edge replaced by $e$.
                This edge may be replaced in $\graphseq'$ at times different from the time at which $e$ was inserted.
                However, this contributes at most $w(e_H) \leq W$ to $\Gamma_2$.
        In case (3), the weight that $e$ contributes to the minimum spanning tree in $\graphseq$ can be reduced by at most $W-1$.
                However, the edge used to replace $e$ may increase the weight of a minimum spanning tree in $\graphseq'$ by 1.
                Thus, $e$ contributes up to $W$ to $\Gamma_2$.
        
        With $|E^*| \leq D$, we thus have $\Gamma_2 \leq DW$.
        This implies
        \begin{align*}
                \GS(\Delta\wmst) &= \sum_{t=1}^T |\Delta\wmst_\graphseq(t) - \Delta\wmst_{\graphseq'}(t)| \\
                                 &= |\Delta\wmst_\graphseq(1) - \Delta\wmst_{\graphseq'}(1)| + \Gamma_2 \\
                                 &\leq 2DW. \qedhere
        \end{align*}
\end{proof}
} \opt{full}{\mstnodegsproof}

\section{Upper Bound for Monotone Functions}
\label{sec:monotone}

In \cref{sec:bounds-global-sensitivity}, we show that privately releasing the difference sequence of a graph sequence does not lead to good error guarantees for partially dynamic problems like minimum cut. Intuitively, the reason is that even for neighboring graph sequences $\graphseq, \graphseq'$, the differences of the difference sequence can be non-zero for all graphs $G_i, G'_i$. In other words, the difference of objective values for the graphs $G_i$ and $G'_i$ can constantly fluctuate during continual updates. However, the difference of objective values, regardless of fluctuations, is \emph{always small}. We show that, by allowing an arbitrarily small \emph{multiplicative} error, we can leverage this fact for a broad class of partially dynamic problems. In particular, we prove that there exist $\eps$-differentially private algorithms for all dynamic problems that are non-decreasing (or non-increasing) on all valid input sequences. This includes, e.g., minimum cut, maximum matching and densest subgraph on partially dynamic inputs. See \cref{alg:monotone} for the details and \cref{tbl:monotone-results} on page~\pageref{tbl:monotone-results} for explicit upper bounds for applications. We state the result for monotonically increasing functions, but it is straightforward to adapt the algorithm to monotonically decreasing functions.\opt{confpre}{ The proof is in \cref{sec:monotone-proof}.}

\SetKwFunction{FnInitMon}{Initialize}
\SetKwFunction{FnMon}{Process}
\begin{algorithm}
        \Fn{\FnInitMon{$\mathcal{D}, \rho, \epsilon, r, \beta$}}{
                $k_0 \gets 0$ \;
                \FnInitSvt($\mathcal{D}, \rho, \epsilon, \log_{1+\beta}(r)$) \tcp*{see \cref{alg:svt}}
        }
	\Fn{\FnMon{$f_i$}}{
                $k_i \gets k_{i-1}$ \;
                \While(\tcp*[f]{see \cref{alg:svt}}){$\FnSvt(f_i, (1+\beta)^{k_i}) = \top$}{
                        $k_i \gets k_i + 1$ \;
                }
                \Return $(1+\beta)^{k_i}$ \;
	}
	\caption{\label{alg:monotone} Multiplicative error algorithm for monotone functions}
\end{algorithm}

\begin{theorem}[restate=TechnicalMonotone]
        \label{thm:technical-monotone}
        Let $r > 0$ and let $f$ be any monotonically increasing function on dynamic inputs (e.g., graphs) with range $[1, r]$ and static global sensitivity $\rho := \GSstatic(f)$, where $r$ and $\rho$ are publicly known parameters. Let $\beta \in (0,1), \delta > 0$ and let $\alpha = 16 \log_{1+\beta}(r) \rho \cdot \ln(2T / \delta) / \eps$. There exists an $\eps$-differentially private algorithm for computing $f$ with multiplicative error $(1+\beta)$, additive error $\alpha$ and failure probability $\delta$.
\end{theorem}
\newcommand{\MonotoneProof}{
\begin{proof}
        We reduce the dynamic problem to a problem that can be solved using the sparse vector technique: Assume that we are given the whole sequence of dynamic updates as a database $\mathcal{D}$ in advance. We will remove this assumption later. 
        
        Formally, let $\mathcal{D} := \graphseq = (G_1, \ldots, G_T)$ be the input to the dynamic problem and let $f_i$ be the function that evaluates $f$ on the $i$-th entry $G_i$ (e.g., graph) of such a database. To initialize our algorithm, we run $\FnInitMon(\mathcal{D}, \rho, \epsilon, r, \beta)$. Then we call $\FnMon(f_i)$ for all $f_i$, $i \in [T]$. We claim that the output of $\FnMon(f_i)$ is an approximation to $f_i(\mathcal{D}) = f(G_i)$ with multiplicative error $(1+\beta)$, additive error $\alpha$ and failure probability $\delta$ over the coins of the whole algorithm, and that the algorithm is $\epsilon$-differentially private.

        \textbf{\sffamily Error.}  We condition on the event that all noises in \cref{alg:svt} are small, i.e., $\lvert \zeta \rvert < \alpha / 4$ and, for all $i$, $\lvert \nu_i \rvert < \alpha / 4$. It follows from the definition of the Laplace distribution and a union bound over all queries that this event occurs with probability at least $1 - e^{- \alpha \eps_1 / (4\rho)} - T \cdot e^{- \alpha \eps_2 / (8 \log_{1+\beta}(r) \rho)} \geq 1 - \delta$. Consider query $f_i(\mathcal{D})$ and let $f_i' := f_i(\mathcal{D}) + \nu_i - \zeta$. After the while-loop of $\FnMon(f_i)$, we know that $f_i(\mathcal{D}) - \alpha/2 \leq (1+\beta)^{k_i}$ and $f_i(\mathcal{D}) + \alpha/2 \geq f_{i-1}(\mathcal{D}) + \alpha/2 \ge (1+\beta)^{k_i-1}$ due to the monotonicity of $f_1(\mathcal{D}), ..., f_T(\mathcal{D})$. We analyze the return value of $\FnMon(f_i)$ and show that it is a $(1+\beta)$ approximation with additive error~$\alpha$:
        \begin{gather*}
                f_i(\mathcal{D}) - \alpha
                        \leq \left(\left( 1+\beta \right)^{k_i} + \frac{\alpha}{2} \right) - \alpha
                        \leq \left( 1+\beta \right)^{k_i} \\
                (1+\beta)f_i(\mathcal{D}) + \alpha
                        \geq (1+\beta) \left(\left( 1+\beta \right)^{k_i-1} - \frac{\alpha}{2} \right) + \alpha
                        \geq \left(1+\beta \right)^{k_i}.
        \end{gather*}
        Since $f(G_T) \leq r = (1+\beta)^{\log_{1+\beta}(r)}$, it follows that the algorithm has a multiplicative error $(1+\beta)$ and an additive error of at most $\alpha$ with probability at least $1 - \delta$.
        
        \textbf{\sffamily Privacy.} The privacy follows directly from \cref{thm:svt} because the return value $(1+\beta)^{k_i}$ is derived from the answers of \cref{alg:svt}, which is $\epsilon$-private, and public knowledge (i.e., the fact that $k_0 = 0$).

        \textbf{\sffamily Removing the assumption.} Observe that to answer query $f_i$ on $\mathcal{D}$, \cref{alg:svt} only needs access to $G_i$ or, equivalently, access to the first $i$ updates of the dynamic update sequence. In other words, $G_i$ is read for the first time after answering query $f_{i-1}$. Therefore, we may ask all previous queries $f_j$, $j < i$, even if $G_i$ is not available yet, e.g., when processing update $j$. For dynamic graph sequences and edge-adjacency or node-adjacency, any pair of same-length prefixes of two graph sequences $\graphseq_1, \graphseq_2$ is neighboring if $\graphseq_1$ and $\graphseq_2$ are neighboring (see \cref{def:edge-adjacency,def:node-adjacency}). Therefore, privacy is not affected by making $G_i$ only available to \cref{alg:monotone} just before the $i$-th query, as required by the dynamic setting.
\end{proof}
} \opt{full}{\MonotoneProof}

\section{Lower Bounds for Event-Level Privacy}
\label{sec:edge-dp-lower-bounds}
\opt{full}{
In this section we show lower bounds for the error of edge-differentially private algorithms in the partially dynamic setting.
We consider the problem of releasing the weight of a minimum spanning tree, the minimum cut and maximum weight matching.
To derive the bounds we reduce differentially private counting in binary streams to these problems and apply a lower bound of Dwork \etal\ \cite{dwork10}, which we restate here.
}
\opt{conf}{
We can show lower bounds for the error of edge- and node-differentially private algorithms in the partially dynamic setting.
We derive the bounds by reducing differentially private counting in binary streams to these problems and apply a lower bound of Dwork \etal\ \cite{dwork10}, which we restate here.
}
\begin{theorem}[Lower bound for counting in binary streams \cite{dwork10}]
        \label{thm:counting-lower-bound}
        Any differentially private event-level algorithm for counting over $T$ rounds must have error $\Omega(\log T)$ (even with $\eps = 1$).
\end{theorem}
\opt{conf}{
We obtain a lower bound of $\Omega(W\log T)$ for minimum cut, maximum weighted matching and minimum spanning tree.
The same approach yields a lower bound of $\Omega(\log T)$ for the subgraph-counting problems, counting the number of high-degree nodes and the degree histogram.
See \cref{tbl:sensitivity-results} on page~\pageref{tbl:sensitivity-results}.
}Note that any lower bound for the incremental setting can be transferred to the decremental setting, using the same reductions but proceeding in reverse.
\opt{confpre}{Details and proofs are provided in \cref{sec:edge-dp-lower-bounds-proofs}.
}

\newcommand{\inclowerboundsection}{
\subsection{Minimum Spanning Tree}
We show a lower bound for the error in $\eps$-edge-differentially private estimation of the weight of a minimum spanning tree when the weights are bounded by $W$.

\begin{lemma}
        \label{lem:counting-mst-redux}
        For any pair of adjacent binary streams $\sigma$, $\sigma'$ of length $T$
        there exists a pair of incremental edge-weighted graph sequences $\graphseq$, $\graphseq'$ that are edge-adjacent,
        such that the weights of minimum spanning trees in $\graphseq$ and $\graphseq'$ are
        \begin{equation*}
                \wmst(G_t) = W \sum_{i=1}^t \sigma(i) + T
        \end{equation*}
        and
        \begin{equation*}
                \wmst(G_t') = W \sum_{i=1}^t \sigma'(i) + T,
        \end{equation*}
        respectively, at all time steps $t \in \{1,\dots,T\}$,
        where edge weights are in $\{1,\dots,W\}$.
\end{lemma}
\begin{proof}
        We construct graph sequences on the node set 
        \begin{equation}
                \label{eq:mst-v}
                V = \{0,\dots,T\} \cup \{T+1,\dots,2T+1\}.
        \end{equation}
        The nodes $1,\dots,T$ will form a path in these graph sequences.
        To achieve this we use the edge set $E_0 = \{\{i-1, i\} \mid i = 1,\dots,T\}$ with weight 1 for each edge and initial the graph sequences with the graph $G_0 = (V,E_0)$.
        The edges in $E$ form a minimum spanning tree of $G_0$ of weight $T$.
        
        For any two adjacent binary streams $\sigma$, $\sigma'$ we construct corresponding graph sequences $\graphseq$, $\graphseq'$ with initial graph $G_0$.
        At all times the node set of the graphs in $\graphseq$ and $\graphseq'$ is $V$, as defined in \eqref{eq:mst-v}.
        We now define for each time $t$ the sets of edge-insertions for $\graphseq$ and $\graphseq'$.
        Let $e_0^t = \{t, (t+2) \bmod T\}$ with $w(e_0^t) = W$ and $e_1^t = \{T+t, T+t+1\}$ with $w(e_1^t) = W$.
        At time $t$ we insert
        \begin{equation*}
                \begin{aligned}
                        \Eins{t} &= \begin{cases}
                                        \{e_0^t\}               &\text{ if } \sigma(t) = 0, \\
                                        \{e_0^t, e_1^t\}        &\text{ if } \sigma(t) = 1, \\
                                    \end{cases} \\
                        \Eins{t}' &= \begin{cases}
                                        \{e_0^t\}               &\text{ if } \sigma'(t) = 0, \\
                                        \{e_0^t, e_1^t\}        &\text{ if } \sigma'(t) = 1, \\
                                    \end{cases}
                \end{aligned}
        \end{equation*}
        into $\graphseq$ and $\graphseq'$, respectively.
        Note that $\Eins{t} \neq \Eins{t}'$ only if $\sigma(t) \neq \sigma'(t)$.
        Thus, $\graphseq$ and $\graphseq'$ only differ in the insertion of a single edge and are adjacent.
        
        Inserting an edge $e_0^t$ never changes the weight of a minimum spanning tree, since the edges $e^i$ defined above make up a spanning tree of the nodes $\{1,\dots,T\}$ and have lower weight.
        Inserting an edge $e_1^t$ increases the weight of the minimum spanning tree by $W$.
        Thus, at any time $t$, the weights of the minimum spanning trees in $\graphseq$ and $\graphseq'$ are
        \begin{equation*}
                \wmst(G_t) = W \sum_{i=1}^t \sigma(i) + T
        \end{equation*}
        and
        \begin{equation*}
                \wmst(G_t') = W \sum_{i=1}^t \sigma'(i) + T,
        \end{equation*}
        respectively, which concludes the proof.
\end{proof}
\begin{lemma}
        \label{lem:counting-mst-redux-node}
        For any pair of adjacent binary streams $\sigma$, $\sigma'$ of length $T$
        there exists a pair of incremental edge-weighted graph sequences $\graphseq$, $\graphseq'$ that are node-adjacent,
        such that the weights of minimum spanning trees in $\graphseq$ and $\graphseq'$ are
        \begin{equation*}
                \wmst(G_t) = W \sum_{i=1}^t \sigma(i)
        \end{equation*}
        and
        \begin{equation*}
                \wmst(G_t') = W \sum_{i=1}^t \sigma'(i),
        \end{equation*}
        respectively, at all time steps $t \in \{1,\dots,T\}$,
        where edge weights are in $\{1,\dots,W\}$.
\end{lemma}
\begin{proof}
        Let $G_0 = (v_0, \emptyset)$.
        Let $\graphseq$, $\graphseq'$ be graph sequences with initial graph $G_0$.

        At time $t$ we insert the node sets
        \begin{equation}
                \begin{aligned}
                        \Vins{t} &= \begin{cases}
                                        \{u_t\}             &\text{ if } \sigma(t) = 0, \\
                                        \{u_t, v_t\}        &\text{ if } \sigma(t) = 1, \\
                                    \end{cases} \\
                        \Vins{t}' &= \begin{cases}
                                        \{u_t\}             &\text{ if } \sigma'(t) = 0, \\
                                        \{u_t, v_t\}        &\text{ if } \sigma'(t) = 1, \\
                                    \end{cases}
                \end{aligned}
        \end{equation}
        into $\graphseq$ and $\graphseq'$, respectively.
        When inserting $v_t$ we also insert the edge $\{v_0,v_t\}$ with weight $W$.
        The nodes $u_t$ stay isolated throughout the graph sequences.
        The weight of the minimum spanning tree is the number of nodes $v_t$ that were inserted into the graph sequence, multiplied by the maximum edge weight $W$.
        That is, at any time $t$, the weights of the minimum spanning trees in $\graphseq$ and $\graphseq'$ are
        \begin{equation*}
                \wmst(G_t) = W \sum_{i=1}^t \sigma(i)
        \end{equation*}
        and
        \begin{equation*}
                \wmst(G_t') = W \sum_{i=1}^t \sigma'(i)
        \end{equation*}
        respectively, which concludes the proof.
\end{proof}

\begin{theorem}
        Any $\eps$-edge-differentially private algorithm to compute the weight of a minimum spanning tree
        in an incremental graph sequence of length $T$ must have additive error $\Omega(W \log T)$
        when the edge weights are in $\{1,\dots,W\}$ and $W$ is a publicly known parameter. This holds even with $\eps = 1$.
\end{theorem}
\begin{proof}
        Suppose there exists an $\eps$-edge-differentially private algorithm to compute the weight of a minimum spanning tree
        in an incremental graph sequence with error $o(W \log T)$ at all time steps.
        Using this algorithm and the reduction from \cref{lem:counting-mst-redux}, we can compute the sum of any binary stream
        with error $o(\log T)$ while preserving $\eps$-differential privacy.
        This contradicts \cref{thm:counting-lower-bound}.
\end{proof}

\begin{theorem}
        Any $\eps$-node-differentially private algorithm to compute the weight of a minimum spanning tree
        in an incremental graph sequence of length $T$ must have additive error $\Omega(W \log T)$
        when the edge weights are in $\{1,\dots,W\}$ and $W$ is a publicly known parameter. This holds even with $\eps = 1$.
\end{theorem}
\begin{proof}
        Suppose there exists an $\eps$-node-differentially private algorithm to compute the weight of a minimum spanning tree
        in an incremental graph sequence with error $o(W \log T)$ at all time steps.
        Using this algorithm and the reduction from \cref{lem:counting-mst-redux-node}, we can compute the sum of any binary stream
        with error $o(\log T)$ while preserving $\eps$-differential privacy.
        This contradicts \cref{thm:counting-lower-bound}.
\end{proof}

\subsection{Weighted Minimum Cut}
We can prove a similar lower bound for edge-differentially private weighted minimum cut in incremental graph sequences.
We again reduce differentially private counting in binary streams to differentially private minimum cut in incremental graph sequences.

\begin{lemma}
        \label{lem:counting-cut-redux}
        For any pair of adjacent binary streams $\sigma$, $\sigma'$ of length $T$
        there exists a pair of incremental edge-weighted graph sequences $\graphseq$, $\graphseq'$ that are edge-adjacent,
        such that the weights of the minimum cut in $\graphseq$ and $\graphseq'$ are
        \begin{equation*}
                \wcut(G_t) = W \sum_{i=1}^t \sigma(i)
        \end{equation*}
        and
        \begin{equation*}
                \wcut(G_t') = W \sum_{i=1}^t \sigma'(i)
        \end{equation*}
        respectively, at all time steps $t \in \{1,\dots,T\}$,
        where edge weights are in $\{1,\dots,W\}$.
\end{lemma}
\begin{proof}
        At the core of our reduction is a $(T+1)$-edge-connected graph with sufficiently many missing edges.
        We use the $(T+1)$-dimensional hypercube graph $Q = (V_Q,E_Q)$.
        Each node $v \in V_Q$ is associated with a vector $\mathbf{v} \in \{0,1\}^{T+1}$.
        We associate every time step $t \in \{1,\dots,T\}$ with two nodes in $Q$, which correspond to the following vectors:
        \begin{equation*}
                \begin{aligned}
                        b_t                           &= (b_1,\dots,b_T,0) \\
                        \text{and}\quad \widehat{b_t} &= (1-b_1,\dots,1-b_T,1), \\
                \end{aligned}
        \end{equation*}
        where $b_1,\dots,b_T \in \{0,1\}$ such that $t = \sum_{i=1}^T b_i \cdot 2^{i-1}$.
        Since $b_t$ and $\widehat{b_t}$ differ in at least 2 elements, they are not connected by an edge in $Q$.

        In the following, every edge has weight $W$, including the edges in $Q$.
        Let $G_0 = (\{v^*\},\emptyset) \cup Q$. $G_0$ has a minimum cut of weight 0, since $v^*$ is isolated.
        Let $\graphseq$, $\graphseq'$ be incremental graph sequences with initial graph $G_0$.
        For every time step we define the edges $e_0^t = \{b_t, \widehat{b_t}\}$ and $e_1^t = \{v^*, b_t\}$.
        At time $t$ we insert
        \begin{equation*}
                \begin{aligned}
                        \Eins{t} &= \begin{cases}
                                        \{e_0^t\}               &\text{ if } \sigma(t) = 0, \\
                                        \{e_0^t, e_1^t\}        &\text{ if } \sigma(t) = 1, \\
                                    \end{cases} \\
                        \Eins{t}' &= \begin{cases}
                                        \{e_0^t\}               &\text{ if } \sigma'(t) = 0, \\
                                        \{e_0^t, e_1^t\}        &\text{ if } \sigma'(t) = 1, \\
                                    \end{cases}
                \end{aligned}
        \end{equation*}
        into $\graphseq$ and $\graphseq'$, respectively.
        Note that these graph sequences differ in exactly one edge exactly when $\sigma(t) \neq \sigma'(t)$.
        
        Denote by $\deg_\graphseq(v^*, t)$ the degree of $v^*$ in $\graphseq$ at time $t$.
        We observe that the degree of $v^*$ increases by 1 in the graph sequence if and only if there is a 1 in the corresponding binary stream. Otherwise, its degree does not change.
        The minimum cut is always the trivial cut around $v^*$, since the degree of $v^*$ never exceeds $T$ and $Q$ is $(T+1)$-edge-connected.
        Thus, the weight of the minimum cut at time $t$ is
        \begin{equation*}
                \wcut(G_t) = W \sum_{i=1}^t \sigma(i)
        \end{equation*}
        and
        \begin{equation*}
                \wcut(G_t') = W \sum_{i=1}^t \sigma'(i),
        \end{equation*}
        which concludes the proof.
\end{proof}
\begin{lemma}
        \label{lem:counting-cut-redux-node}
        For any pair of adjacent binary streams $\sigma$, $\sigma'$ of length $T$
        there exists a pair of incremental edge-weighted graph sequences $\graphseq$, $\graphseq'$ that are node-adjacent,
        such that the weights of the minimum cut in $\graphseq$ and $\graphseq'$ are
        \begin{equation*}
                \wcut(G_t) = W \sum_{i=1}^t \sigma(i)
        \end{equation*}
        and
        \begin{equation*}
                \wcut(G_t') = W \sum_{i=1}^t \sigma'(i)
        \end{equation*}
        respectively, at all time steps $t \in \{1,\dots,T\}$,
        where edge weights are in $\{1,\dots,W\}$.
\end{lemma}
\begin{proof}
        Let $G_0 = (\{u_0,v_0\}, \emptyset)$.
        Let $\graphseq$, $\graphseq'$ be graph sequences with initial graph $G_0$.
        In the following all edges have weight $W$.
       
        At time $t$ we insert the node sets
        \begin{equation*}
                \begin{aligned}
                        \Vins{t} &= \begin{cases}
                                        \{u_t\}             &\text{ if } \sigma(t) = 0, \\
                                        \{u_t, v_t\}        &\text{ if } \sigma(t) = 1, \\
                                    \end{cases} \\
                        \Vins{t}' &= \begin{cases}
                                        \{u_t\}             &\text{ if } \sigma'(t) = 0, \\
                                        \{u_t, v_t\}        &\text{ if } \sigma'(t) = 1, \\
                                    \end{cases}
                \end{aligned}
        \end{equation*}
        into $\graphseq$ and $\graphseq'$, respectively.
        Along with every node $u_t$, we insert the edges $\{u_t,u_j\}$ for all $0 \leq j < t$.
        If we insert $v_t$, we also insert the edges $\{v_t,v_j\}$ for all $0 \leq j < t$ and the edge $\{v_t,u_0\}$.
        
        Let $U_t = \{u_0,\dots,u_t\}$ and $V_t = \{v_0,\dots,v_t\}$.
        The subgraphs induced by $U_t$ and $V_t$ are always complete.
        The cut $(U_t,V_t)$ always has minimum weight.
        The weight of this cut in $\graphseq$ increases by $W$ only if $\sigma(t) = 1$; otherwise, it does not change.
        The same holds for $\graphseq'$ and $\sigma'$.
        Thus, the weight of the minimum cut at time $t$ is
        \begin{equation*}
                \wcut(G_t) = W \sum_{i=1}^t \sigma(i)
        \end{equation*}
        and
        \begin{equation*}
                \wcut(G_t') = W \sum_{i=1}^t \sigma'(i),
        \end{equation*}
        which concludes the proof.
\end{proof}

\begin{theorem}
        Any $\eps$-edge-differentially private algorithm to compute the weight of a minimum cut
        in an incremental graph sequence of length $T$ must have additive error $\Omega(W \log T)$
        when the edge weights are in $\{1,\dots,W\}$ and $W$ is a publicly known parameter. This holds even with $\eps = 1$.
\end{theorem}
\begin{proof}
        Suppose there exists an $\eps$-edge-differentially private algorithm to compute the weight of a minimum cut
        in an incremental graph sequence with error $o(W \log T)$ at all time steps.
        Using this algorithm and the reduction from \cref{lem:counting-cut-redux}, we can compute the sum of any binary stream
        with error $o(\log T)$ while preserving $\eps$-differential privacy.
        This contradicts \cref{thm:counting-lower-bound}.
\end{proof}

\begin{theorem}
        Any $\eps$-node-differentially private algorithm to compute the weight of a minimum cut
        in an incremental graph sequence of length $T$ must have additive error $\Omega(W \log T)$
        when the edge weights are in $\{1,\dots,W\}$  and $W$ is a publicly known parameter. This holds even with $\eps = 1$.
\end{theorem}
\begin{proof}
        Suppose there exists an $\eps$-node-differentially private algorithm to compute the weight of a minimum cut
        in an incremental graph sequence with error $o(W \log T)$ at all time steps.
        Using this algorithm and the reduction from \cref{lem:counting-cut-redux-node}, we can compute the sum of any binary stream
        with error $o(\log T)$ while preserving $\eps$-differential privacy.
        This contradicts \cref{thm:counting-lower-bound}.
\end{proof}

\subsection{Maximum Weighted Matching}
\begin{lemma}
        \label{lem:counting-match-redux}
        For any pair of adjacent binary streams $\sigma$, $\sigma'$ of length $T$
        there exists a pair of incremental edge-weighted graph sequences $\graphseq$, $\graphseq'$ that are edge-adjacent,
        such that the weights of the maximum weighted matching in $\graphseq$ and $\graphseq'$ are
        \begin{equation*}
                \wmatch(G_t) = W \sum_{i=1}^t \sigma(i) + W\cdot T
        \end{equation*}
        and
        \begin{equation*}
                \wmatch(G_t') = W \sum_{i=1}^t \sigma'(i) + W\cdot T,
        \end{equation*}
        respectively, at all time steps $t \in \{1,\dots,T\}$,
        where edge weights are in $\{1,\dots,W\}$.
\end{lemma}
\begin{proof}
        Let $G_0 = (U \cup V, E)$, where $U = \{u_1,\dots,u_{2\cdot T}\}$, $V = \{v_1,\dots,v_{2\cdot T}\}$ and
        $E = \bigcup_{i=1}^T \{u_i, u_{(i+1) \mod T}\}$.
        All edges have weight $W$.
        Let $\graphseq$ and $\graphseq'$ be graph sequences with initial graph $G_0$.

        For all times $t \in \{1,\dots,T\}$ we define the edges $e_0^t = \{u_t,v_t\}$ with $w(e_0^t) = 1$ and $e_1^t = \{u_{T+t}, v_{t+T}\}$ with $w(e_1^t) = W$.
        At time $t$ we insert
        \begin{equation*}
                \begin{aligned}
                        \Eins{t} &= \begin{cases}
                                        \{e_0^t\}               &\text{ if } \sigma(t) = 0, \\
                                        \{e_0^t, e_1^t\}        &\text{ if } \sigma(t) = 1, \\
                                    \end{cases} \\
                        \Eins{t}' &= \begin{cases}
                                        \{e_0^t\}               &\text{ if } \sigma'(t) = 0, \\
                                        \{e_0^t, e_1^t\}        &\text{ if } \sigma'(t) = 1, \\
                                    \end{cases}
                \end{aligned}
        \end{equation*}
        into $\graphseq$ and $\graphseq'$, respectively.
        Note that these graph sequences differ in exactly one edge exactly when $\sigma(t) \neq \sigma'(t)$.

        $G_0$ has a matching of weight $W\cdot T$, which consists of the edges in $E$.
        Inserting an edge $e_0^t$ for any $t$ does not change the weight of the maximum matching. 
        Inserting an edge $e_1^t$ however increases the weight of the matching by $W$, since two previously isolated nodes are now connected by this new edge.
        Thus, the weight of the maximum matching at time $t$ is
        \begin{equation*}
                \wmatch(G_t) = W \sum_{i=1}^t \sigma(i) + W\cdot T
        \end{equation*}
        and
        \begin{equation*}
                \wmatch(G_t') = W \sum_{i=1}^t \sigma'(i) + W\cdot T,
        \end{equation*}
        which concludes the proof.
\end{proof}

We can also reduce differentially private binary counting to node-differentially private maximum weight matching in incremental graph sequences.
\begin{lemma}
        \label{lem:counting-match-redux-node}
        For any pair of adjacent binary streams $\sigma$, $\sigma'$ of length $T$
        there exists a pair of incremental edge-weighted graph sequences $\graphseq$, $\graphseq'$ that are node-adjacent,
        such that the weights of the maximum weighted matching in $\graphseq$ and $\graphseq'$ are
        \begin{equation*}
                \wmatch(G_t) = W \sum_{i=1}^t \sigma(i)
        \end{equation*}
        and
        \begin{equation*}
                \wmatch(G_t') = W \sum_{i=1}^t \sigma'(i)
        \end{equation*}
        respectively, at all time steps $t \in \{1,\dots,T\}$,
        where edge weights are in $\{1,\dots,W\}$.
\end{lemma}
\begin{proof}
        Let $G_0 = (\{v_1,\dots,v_T\}, \emptyset)$.
        Let $\graphseq$, $\graphseq'$ be graph sequences with initial graph $G_0$.
        
        At time $t$, we insert a node $u_t$ into $\graphseq$ and $\graphseq'$.
        If $\sigma(t) = 1$ we additionally insert node $v_t'$ along with the edge $\{v_t', v_t\}$ with weight $W$ into $\graphseq$.
        Similarly, if $\sigma'(t) = 1$ we insert node $v_t'$ along with the edge $\{v_t', v_t\}$ with weight $W$ into $\graphseq'$.

        Then, the weight of the maximum matching in $\graphseq$ at time $t$ increases by $W$ if $\sigma(t) = 1$; otherwise, it does not change.
        The same holds for $\sigma'$ and $\graphseq'$.
        Thus, the weight of the maximum matching at time $t$ is
        \begin{equation*}
                \wmatch(G_t) = W \sum_{i=1}^t \sigma(i)
        \end{equation*}
        and
        \begin{equation*}
                \wmatch(G_t') = W \sum_{i=1}^t \sigma'(i),
        \end{equation*}
        which concludes the proof.
\end{proof}

\begin{theorem}
        Any $\eps$-edge-differentially private algorithm to compute the weight of a maximum weight matching
        in an incremental graph sequence of length $T$ must have additive error $\Omega(W \log T)$
        when the edge weights are in $\{1,\dots,W\}$ and $W$ is a publicly known parameter. This holds even with $\eps = 1$.
\end{theorem}
\begin{proof}
        Suppose there exists an $\eps$-edge-differentially private algorithm to compute the weight of a maximum weight matching
        in an incremental graph sequence with error $o(W \log T)$ at all time steps.
        Using this algorithm and the reduction from \cref{lem:counting-match-redux}, we can compute the sum of any binary stream
        with error $o(\log T)$ while preserving $\eps$-differential privacy.
        This contradicts \cref{thm:counting-lower-bound}.
\end{proof}
\begin{theorem}
        Any $\eps$-node-differentially private algorithm to compute the weight of a maximum weight matching
        in an incremental graph sequence of length $T$ must have additive error $\Omega(W \log T)$
        when the edge weights are in $\{1,\dots,W\}$ and $W$ is a publicly known parameter. This holds even with $\eps = 1$.
\end{theorem}
\begin{proof}
        Suppose there exists an $\eps$-node-differentially private algorithm to compute the weight of a maximum weight matching
        in an incremental graph sequence with error $o(W \log T)$ at all time steps.
        Using this algorithm and the reduction from \cref{lem:counting-match-redux-node}, we can compute the sum of any binary stream
        with error $o(\log T)$ while preserving $\eps$-differential privacy.
        This contradicts \cref{thm:counting-lower-bound}.
\end{proof}

\subsection{Subgraph Counting, High-Degree Nodes and Degree Histogram}
We show that, for the number of high-degree nodes, the degree histogram and the subgraph counting problems adjacent binary streams
can be reduced to adjacent incremental graph sequences. This implies a lower bound of $\Omega(\log T)$ on the error for these problems.

\begin{lemma}
        \label{lem:counting-reductions}
        There exist functions
        $\mathscr{G}_\tau$,
        $\mathscr{G}_h$,
        $\mathscr{G}_\Delta$,
        $\mathscr{G}_k$,
        that map a binary stream to an incremental graph sequence of the same length,
        such that the following holds:
        At time $t$
        \begin{enumerate}
                \item the number of nodes of degree at least $\tau$ in $\mathscr{G}_\tau(\sigma)$,
                \item the number of nodes of degree 2 in $\mathscr{G}_h(\sigma)$,
                \item the number of triangles in $\mathscr{G}_\Delta(\sigma)$ and
                \item the number of $k$-stars in $\mathscr{G}_k(\sigma)$
        \end{enumerate}
        is equal to $\sum_{i=1}^{t} \sigma(i)$.
        Furthermore, all functions map adjacent binary streams to edge-adjacent graph sequences.
\end{lemma}
\begin{proof}
        Let $\sigma$ be a binary stream.
        \begin{enumerate}
\item $\mathscr{G}_\tau$ outputs a graph sequence with initial graph $G_\tau = \cup_{i=1}^T S_i$.
                        Each $S_i$ is a $\tau$-star with one missing edge.
                        At time $t$, if $\sigma(t) = 1$, then the missing edge is inserted into $S_t$.

\item $\mathscr{G}_h$ outputs a graph sequence on the nodes $V_h = \cup_{i=1}^T \{a_i,b_i,c_i\}$.
                        Initially, there are edges $\{b_i,c_i\}$ for each $i \in \{1,\dots,T\}$.
                        At time $t$, if $\sigma(t) = 1$, then the edge $\{a_t,b_t\}$ is inserted.
                        Note that $\deg(b_t) = \sigma(t) + 1$ and $\deg(a_i),\deg(c_i) \leq 1$ for all $t$.

\item $\mathscr{G}_\Delta$ outputs a graph sequence on the nodes $V_\Delta = \cup_{i=1}^T \{a_i,b_i,c_i\}$.
                        Initially, there are edges $\{a_i,b_i\}$ and $\{b_i,c_i\}$ for each $i \in \{1,\dots,T\}$.
                        At time $t$, if $\sigma(t) = 1$, then the edge $\{a_t,c_t\}$ is inserted.
\item $\mathscr{G}_k$ outputs a graph sequence with initial graph $G_k = \cup_{i=1}^T S_i$.
                        Each $S_i$ is a $k$-star with one missing edge.
                        At time $t$, if $\sigma(t) = 1$, then the missing edge is inserted into $S_t$.
        \end{enumerate}
        In all graph sequences no edge is inserted if $\sigma(t) = 0$; if $\sigma(t) = 1$, then a single edge is inserted.
        Thus, the functions map adjacent binary streams to edge-adjacent graph sequences.
\end{proof}
\begin{lemma}
        \label{lem:counting-reductions-node}
        There exist functions
        $\mathscr{G}_\tau'$,
        $\mathscr{G}_h'$,
        $\mathscr{G}_e'$,
        $\mathscr{G}_\Delta'$,
        $\mathscr{G}_k'$,
        that map a binary stream to an incremental graph sequence of the same length with maximum degree $D > 3$, $\tau,k \leq D$,
        such that the following holds:
        Let $\sigma$ be a binary stream.
        At time $t$
        \begin{enumerate}
                \item the number of nodes of degree at least $\tau$ in $\mathscr{G}_\tau'(\sigma)$,
                \item the number of nodes of degree 2 in $\mathscr{G}_h'(\sigma)$,
                \item the number of edges in $\mathscr{G}_e'(\sigma)$,
                \item the number of triangles in $\mathscr{G}_\Delta'(\sigma)$,
                \item the number of $k$-stars in $\mathscr{G}_k'(\sigma)$
        \end{enumerate}
        is equal to $D\sum_{i=1}^t \sigma(i)$.
        Furthermore, all functions map adjacent binary streams to node-adjacent graph sequences.
\end{lemma}
\begin{proof}
        Let $\sigma$ be a binary stream.
        \begin{enumerate}
\item $\mathscr{G}'_\tau$ outputs a graph sequence with initial graph $G_\tau = \cup_{i=1}^T S_i$.
                        Each $S_i$ is the union of $(D-1)$-many $(\tau-1)$-stars.
                        At time $t$, if $\sigma(t) = 1$, then a node $v_t$ is inserted into $S_t$,
                        along with edges from $v_t$ to the central node of each $(\tau-1)$-star in $S_t$.
                        Since $D \geq \tau$, the number of $\tau$-stars in $S_t$ is $D$ if $\sigma(t) = 1$ and zero otherwise.
                        In $S_t$ the centers of the $\tau$-stars have degree $\geq \tau$, all other nodes have degree $< \tau$.

\item $\mathscr{G}'_h$ outputs a graph sequence with initial graph $G_h = \cup_{i=1}^T M_i$,
                        where $M_i$ consists of $D$ pairwise connected nodes,
                        i.e., $M_i = (\bigcup_{j=1}^D\{a_j,b_j\}, \bigcup_{j=1}^D \{\{a_j,b_j\}\})$.
                        At time $t$, if $\sigma(t) = 1$, then a node $v_t$ is inserted into $M_t$
                        along with edges $\{v_t,a_j\}$ for all $j=1,\dots,D$.
                        The number of nodes of degree 2 is $D$ in $M_t$ if $\sigma(t) = 1$ and zero otherwise.

\item $\mathscr{G}'_e$ outputs a graph sequence with initial graph $G_e = \cup_{i=1}^T V_i$,
                        where $V_i = (\{1,\dots,D\}, \emptyset)$.
                        At time $t$, if $\sigma(t) = 1$, then a node $v_t$ is inserted into $V_t$
                        along with edges from $v_t$ to all nodes in $V_t$.
                        The number of edges is $D$ in $V_t$ if $\sigma(t) = 1$ and zero otherwise.

\item $\mathscr{G}'_\Delta$ outputs a graph sequence with initial $G_\Delta = \cup_{i=1}^T C_i$,
                        where $C_i$ is a cycle of $D$ nodes.                        
                        At time $t$, if $\sigma(t) = 1$, then a node $v_t$ is inserted into $C_t$
                        along with edges from $v_t$ to every node in $C_t$.
                        Thus, if $\sigma(t) = 1$ the number of triangles in $C_t$ is $D$; otherwise, it is zero.

\item $\mathscr{G}'_k$ outputs a graph sequence with initial graph $G_\tau = \cup_{i=1}^T S_i$.
                        Each $S_i$ is the union of $(D-1)$-many $(\tau-1)$-stars.
                        At time $t$, if $\sigma(t) = 1$, then a node $v_t$ is inserted into $S_t$,
                        along with edges from $v_t$ to the central node of each $(\tau-1)$-star in $S_t$.
                        Since $D \geq \tau$, the number of $\tau$-stars in $S_t$ is $D$ if $\sigma(t) = 1$ and zero otherwise.
        \end{enumerate}
        In all graph sequences no node is inserted if $\sigma(t) = 0$; if $\sigma(t) = 1$, then a single node is inserted.
        Thus, the functions map adjacent binary streams to node-adjacent graph sequences.
\end{proof}

\begin{theorem}
        Any $\eps$-edge-differentially private algorithm to compute
        \begin{enumerate}
                \item the number of high-degree nodes,
                \item the degree histogram,
                \item the number of edges,
                \item the number of triangles, or
                \item the number of $k$-stars
        \end{enumerate}
        in an incremental graph sequence of length $T$ must have additive error $\Omega(\log T)$, even with $\eps = 1$.
\end{theorem}
\begin{proof}
        By \cref{lem:counting-reductions}, adjacent binary streams can be encoded in edge-adjacent incremental graph sequences of the same length.
        Counting the number of edges is equivalent to binary counting.
        The claim follows by \cref{thm:counting-lower-bound}.
\end{proof}
\begin{theorem}
        There exists  an incremental graph sequence $\mathcal{G}$ of length $T$ where every node is incident to at most $D$ edges such that
        any $\eps$-node-differentially private algorithm that knows the parameter $D$ and  computes
        \begin{enumerate}
                \item the number of high-degree nodes,
                \item the degree histogram,
                \item the number of edges,
                \item the number of triangles, or
                \item the number of $k$-stars
        \end{enumerate}
        on $\mathcal{G}$  must have additive error $\Omega(D\log T)$. This holds even with $\eps = 1$.
\end{theorem}
\begin{proof}
        By \cref{lem:counting-reductions-node}, adjacent binary streams can be encoded in node-adjacent incremental graph sequences of the same length.
        The claim follows by \cref{thm:counting-lower-bound}.
\end{proof}

} \opt{full}{\inclowerboundsection} 

\section{Lower Bound for User-Level Privacy}
\label{sec:lower-bound-user-level}
\NewDocumentCommand{\trans}{m m}{\operatorname{\tau}(#1, #2)}
\NewDocumentCommand{\transnum}{m m}{\lvert \trans{#1}{#2} \rvert }
\NewDocumentCommand{\transseq}{m m}{\operatorname{T}(#1, #2)}
\NewDocumentCommand{\rev}{m}{{#1}^{-1}
}
\NewDocumentCommand{\eqclas}{m O{\sim}}{[#1]_{\sim}}

We show that for several fundamental problems on dynamic graphs like minimum spanning tree and minimum cut, a differentially private algorithm with edge-adjacency on user-level must have an additive error that is linear in the maximum function value.
Technically, we define the \emph{spread} of a graph function as the maximum difference of the function's value on any two graphs.
Then, we show that any algorithm must have an error that is linear in the graph function's spread.
We write this section in terms of edge-adjacency but the corresponding result for node-adjacency carries over.
See \cref{tbl:sensitivity-results} on page~\pageref{tbl:sensitivity-results} for the resulting lower bounds.

\opt{conf}{
\begin{definition}
Let $G_1$ and $G_2$ be a pair of graphs. We define $\trans{G_1}{G_2}$ to be an update sequence $u_1, \ldots, u_\ell$ of minimum length that transforms $G_1$ into $G_2$.
We denote the graph sequence that results from applying $\trans{G_1}{G_2}$ to $G_1$ by $\transseq{G_1}{G_2}$.

        Let $s, \ell: \nats \rightarrow \{ 2i \mid i \in \nats \}$ be functions. A graph function $f$ has spread $(s(n), \ell(n))$ on inputs of size $n$ if, for every $n$, there exist two graphs $G_1, G_2$ of size $n$ so that $\lvert f(G_1) - f(G_2) \rvert \geq s(n)$
and $\transnum{G_1}{G_2} = \ell(n)$.
\end{definition}
}

\opt{full}{
\begin{definition}[adjacency transformation]
        \label{def:transform}
        Let $G_1$ and $G_2$ be a pair of graphs. We define $\trans{G_1}{G_2}$ to be an update sequence $u_1, \ldots, u_\ell$ of minimum length that transforms $G_1$ into $G_2$.
We denote the graph sequence that results from applying $\trans{G_1}{G_2}$ to $G_1$ by $\transseq{G_1}{G_2}$, i.e., $\transseq{G_1}{G_2} = (G_1, \apply{G_1}{u_1}, \apply{(\apply{G_1}{u_1})}{u_2}, \ldots, G_2)$.
\end{definition}

\begin{definition}[spread]
        \label{def:spread}
        Let $s, \ell: \nats \rightarrow \{ 2i \mid i \in \nats \}$ be functions. A graph function $f$ has spread $(s(n), \ell(n))$ on graphs of size $n$ if, for every $n$, there exist two graphs $G_1, G_2$ of size $n$ so that $\lvert f(G_1) - f(G_2) \rvert \geq s(n)$
and $\transnum{G_1}{G_2} = \ell(n)$. Furthermore, $f$ \emph{spares an edge $e$} if $e \in E(G_1) \cap E(G_2)$.
\end{definition}
}

\opt{confpre}{The proof of the following theorem appears in \cref{sec:lower-bound-proof}.}

\begin{theorem}[restate=TechnicalLowerBound]
        \label{thm:technical-lower-bound}
        Let $\epsilon, \delta > 0$ and let $f$ be a graph function for $n$-node graphs with publicly known spread $(s(n),\ell(n))$ that spares an edge $e$. For streams of length $T$ on graphs of size $n$, where $T > 2\ell \log(e^{4 \epsilon \ell} / (1-\delta)) \in O(\epsilon \ell^2 + \ell\log(1 / (1-\delta)))$, every $\epsilon$-differentially private  dynamic algorithm with user-level edge-adjacency that computes $f$ with probability at least $1-\delta$ must have error $\Omega(s(n))$.
\end{theorem}
\newcommand{\LowerBoundProof}{
\begin{proof}
    For any update $u$, let $\rev{u}$ be the update that reverses $u$.
    For the sake of contradiction, assume that there exists an $\epsilon$-differentially private graph algorithm $\algo$ with edge-adjacency on user-level that with probability at least $1-\delta$ computes $f$ with error less than $s(n) / 2$.
    Let $n >0$, $s := s(n)$, $\ell = \ell(n)$ and let $G_1, G_2$ be graphs of size $n$ with spread $(s,\ell)$ on $f$ that spares an edge $e$.
    We assume below wlog that $\ell$ is even, if it is odd the proof can be easily adapted.
    Let $u$ be the update  operation that inserts $e$ into the current graph (see below).
    Let $(o_0, \ldots, o_{\ell-1}) := \trans{G_1}{G_2}$ and let $(o'_0, \ldots, o'_{\ell-1}) := \trans{G_2}{G_1}$.
    Without loss of generality, assume that $T$ is a multiple of $2 \ell$.
    Let $\mathcal{B}$ be the set of bit strings of length $T / (2\ell)$.
    For every bit string $b = (b_0, \ldots, b_{T / (2\ell) - 1}) \in \mathcal{B}$, we construct a unique graph sequence $\graphseq_b = (H_0, \ldots, H_{T})$ as follows, resulting in $2^{T/(2\ell)}$
    different graph sequences.
     
    The sequence is partitioned into phases of length $2 \ell$.
    Each phase has a \emph{type}: it is either a \emph{forward} phase or a \emph{backward} phase.
    The sequence starts with a forward phase and alternates between forward and backward phases.
    Intuitively, forward phases transform $G_1$ into $G_2$, while backward phases transform $G_2$ into $G_1$. This takes $\ell$ updates.
    To generate an exponential number of different but ``close'' graph sequences each phase also contains $\ell$ \emph{placeholder updates}.
    The purpose of a placeholder update is to basically not modify the graph in any of the graph sequences. Thus the first placeholder update  updates $e$ and the next update is the inverse operation.
    This is repeated $\ell/2$ times.
    For every phase, a corresponding bit in $b$ decides whether the transformation happens \emph{before} the placeholder updates or \emph{after} placeholder updates. 
    Thus, $2^{T/(2\ell)}$ different graph sequences are generated, where the graphs at position
    $2i \ell$ for any integer $i\ge 0$ in all these sequences are identical.
    
    More formally, for integer $i = 0, 1, ..., T/(2\ell) - 1$, the bit $b_i$ corresponds to the phase $H_{2i \ell}, \ldots, H_{2(i+1) \ell - 1}$ in the sequence $\graphseq_b$.
    Bit $b_i$ corresponds to graphs in a forward phase if and only if $i$ is even.
    We define $\graphseq_b$ as follows:
    \begin{description}
        \item[forward phase, $b_i = 0$:] For any $j$, $0 \leq j < \ell$ (i.e., the first half of the phase), we define $H_{2i \ell + j+1} := \apply{H_{2i \ell + j }}{o_j}$.
        The second half of the phase, i.e., $H_{(2i+1) \ell}, \ldots, H_{2(i+1) \ell - 1}$, is defined by alternating between $H_{(2i+1) \ell}$ and $\apply{H_{(2i+1) \ell}}{u}$, which are placeholder updates.
        \item[forward phase, $b_i = 1$:] This phase results from swapping the first and the second half of the forward phase with $b_i=0$.
        In particular, the first half of the phase, i.e., $H_{2i \ell}, \ldots, H_{(2i+1) \ell - 1}$, is defined by alternating between $H_{2i \ell}$ and $\apply{H_{2i \ell}}{u}$, which are placeholder updates.
        For any $j$, $\ell \leq j < 2\ell$, $H_{2i \ell + j+1} := \apply{H_{2i \ell + j}}{o_{j-\ell}}$.
        \item[backward phase, $b_i \in \{ 0,1 \}$:] The graphs $H_{2i \ell}, \ldots, H_{2(i+1)\ell-1}$ are defined analogously to a forward phase with $b_i = 0$ or $b_i = 1$, respectively.
        The only difference is that instead of $o_0, \ldots, o_{\ell-1}$, the updates $o'_0, \ldots, o'_{\ell-1}$ are used.
    \end{description}
    It follows from the construction that for any bit string $b$ and its corresponding graph sequence $(H_0, \ldots, H_{T})$, it holds that $H_i = G_1$ if $i \mod{4\ell} = 0$ and $H_i = G_2$ if $i \mod{4\ell} = 2\ell$. 
    However, for every pair of bit strings $b, b'$, $b \neq b'$, and their corresponding graph sequences $(H_0, \ldots, H_{T-1})$ and $(H'_0, \ldots, H_{T-1})$, there exists an $i$ with $i\mod{2\ell} = \ell$, so that $\lvert f(H_i) - f(H'_i) \rvert \geq s$.
    This stems from the fact that $f$ has spread $(s,\ell)$ and that $b$ and $b'$ must differ in at least one bit, say $b_i = 0$ and $b'_i = 1$.
    Thus, by construction, $H_{(2i+1)\ell} = G_2$ (as $\trans{G_1}{G_2}$ has already been executed in this phase)  and $H'_{(2i+1)\ell} = G_1$ (as $\trans{G_1}{G_2}$ has not yet been executed in this phase).
    
    Next we show that any graph sequence $\graphseq_b$ is on user-level ``edge-close'' to a very generic graph sequence.
    More specifically, let $\graphseq'$ be the sequence of graphs that alternates between $G_1$ and $\apply{G_1}{u}$, i.e., $(G_1, \apply{G_1}{u}, G_1, \apply{G_1}{u},\ldots)$.
    We argue that there is a (short) sequence $\Gamma_1, \Gamma_2, \dots \Gamma_{2\ell+1}$ of  graph sequences
    such that $\Gamma_i$ and $\Gamma_{i+1}$ are \emph{user-level edge-adjacent} graph sequences
    with $\Gamma_1 = \graphseq_b$ and $\Gamma_{2\ell+1} = \graphseq'$.
    In other words, the sequence of graph sequences transforms $\graphseq_b$ into $\graphseq'$.
    Instead of talking about a transforming graph sequence we use below the notation of \emph{user-level edge-adjacency operation}: It takes one operation to transform a graph sequence $\Gamma_i$ into a user-level edge-adjacent graph sequence $\Gamma_{i+1}$.
    
    We now give the details:
    For any $b$, each phase of $\graphseq_b$ can be divided into two halves $(H_{2i\ell}, \ldots, H_{(2i+1)\ell-1})$ and $(H_{(2i+1)\ell}, \ldots H_{(2(i+1)\ell-1)})$: in one half (the first if $b_i = 0$, the latter if $b_i = 1$), the updates are already alternating between $u$, i.e., inserting $e$, and $\rev{u}$, i.e., deleting $e$.
    Therefore, we only need to transform the updates on the other half into a sequence that alternates between $u$ and $\rev{u}$. 
    The other half is either $\transseq{G_1}{G_2}$ (in forward phases) or $\transseq{G_2}{G_1}$ (in backward phases).
    By \cref{lem:edge-adj-ulvl} (see below), there is a sequence of $2\ell + 2$ graph sequences $\graphseq_0',\dots,\graphseq_{2\ell+1}'$ that starts with $\graphseq_0' = \transseq{G_1}{G_2}$ and ends with $\graphseq_{2\ell+1}' = \graphseq'$.
    Furthermore, for each $0 \leq i \leq 2\ell$ there is an edge $e_i$ such that $\graphseq_i'$ and $\graphseq_{i+1}'$ are \emph{event-level} edge-adjacent on $e_i$.
    Similarly, such sequence and edges $(e'_0, \ldots, e'_{2\ell})$ exist for $\transseq{G_2}{G_1}$. 
    Recall that \emph{user-level edge-adjacency} allows to modify a graph sequence by updating an arbitrary number of graphs in the sequence \emph{using the same edge $e$, inserting $e$ into some graphs that do not contain $e$ and removing it from some of the graphs that contain $e$}.
    Also recall that $\graphseq_b$ consists of many repetitions of $\transseq{G_1}{G_2}$.
    Thus one user-level edge-adjacency operation  using the edge $e_1$ allows to modify \emph{all} occurrences of $\transseq{G_1}{G_2}$ to reflect the update of edge $e_1$.
    Afterwards a second user-level edge-adjacency operation using $e_2$ modifies the resulting sequence etc.
    After $2 \ell + 1$ user-level edge-adjacency operations all occurrences of $\transseq{G_1}{G_2}$ have been transformed into sequences that alternate $G_1$ and $ \apply{G_1}{u}$.
    The same can be done for the repetitions of $\transseq{G_2}{G_1}$.
    This shows that at most $2\ell +1$ user-level edge-adjacency operations suffice to transform $\graphseq_b$ into $\graphseq'$.
    
    Recall that we assumed the existence of an $\epsilon$-differentially private graph algorithm $\algo$ on user-level that, with probability at least $1-\delta$, computes $f$ with error less than $s/2$.
    After each update $\algo$ has to output $f$, i.e., the sequence of outputs has the same length as the sequence of updates.
    Let $O_b$ be the set of output sequences of $\cal{A}$ with additive error at most $s/2$ when processing $\graphseq_b$ with $G_1$ as initial graph.
    Now we consider the family of output sets $O_b$ for all $b \in \mathcal{B}$. 
    For any two $b, b' \in \cal{B}$ with $b \ne b'$ consider the value of $f$ applied to the two 
    graph sequences $\graphseq_b$ and $\graphseq_{b'}$ and recall that there is at least one index in these two graph sequences where the $f$-value applied to the two graphs at this index differs by at least $s$.
    Thus, the event that $\algo$ gives an answer with additive error less than $s/2$ on $\graphseq_b$ (i.e., returns a sequence from $O_b$) and  the event hat 
    $\algo$ gives an answer with additive error less than $s/2$ on $\graphseq_{b'}$ (i.e., returns a sequence from $O_{b'}$) are pairwise disjoint.
    However,  for any $b \in \cal{B}$ by $\epsilon$-differential privacy,
    the probability that $\algo$ outputs an sequence from $O_b$ when run on $\graphseq'$, i.e.~$\Pr[\algo(\graphseq') \in  O_b]$, is at least
    $e^{-(2\ell + 1)\epsilon} \cdot \Pr[\algo(\graphseq_b) \in O_b] \geq e^{-4\epsilon \ell} (1-\delta)$. 
    Since $T > 2\ell\log(e^{4 \epsilon \ell} / (1-\delta))$,
    we have that $\lvert \mathcal B \rvert = 2^{T/2\ell} > e^{4 \epsilon \ell} / (1-\delta)$.
    It follows that
    $\Pr[\algo(\graphseq') \in \cup_{b \in \mathcal{B}} O_b] = \sum_{b \in \mathcal{B}} e^{-4 \epsilon \ell} \Pr[\algo(\graphseq_b) \in O_b] \geq |\mathcal{B}|e^{-4\epsilon\ell}(1-\delta) > 1$,
    which is a contradiction.
\end{proof}

\begin{lemma}
        \label{lem:edge-adj-ulvl}
        Let $n > 0$ and let $f$ be a graph function on $n$-node graphs with spread $(s(n),\ell(n))$ that spares an edge $e$ on $G_1, G_2$ according to Definition~\ref{def:transform}. Let $u$ be the update that inserts $e$, and let $\graphseq' := (G_1, \apply{G_1}{u}, G_1, \apply{G_1}{u}, \ldots,)$ of length $\ell + 1 := \transnum{G_1}{G_2} + 1$. Then, there exists a sequence of graph sequences $(\graphseq_0 := \transseq{G_1}{G_2}, \ldots, \graphseq_{2\ell+1} := \graphseq')$ and a sequence of edges $(e_0, e_1, \ldots, e_{2\ell+1})$ so that, for any $0\leq i \leq 2\ell$, $\graphseq_i$ and $\graphseq_{i+1}$ are event-level edge-adjacent on $e_i$.
\end{lemma}
\begin{proof}
    For $i \in \{ 0 \} \cup [2\ell + 1]$, we construct sequences $\graphseq_i = (G_{i,1}, \ldots, G_{i,\ell+1})$.
    Define $\graphseq_0 := \transseq{G_1}{G_2}$.
    We transform $\graphseq_0$ by alternatingly applying $u$ (i.e.~we insert $e$) and $\rev{u}$ (i.e.~we delete $e$)  to its graphs (in addition to the existing updates).
    For $i \in \{1, \ldots, \ell\}$, we define
    \begin{equation*}
            \graphseq_i := (G_{i-1,1}, \ldots,
                            G_{i-1,i}, \apply{G_{i-1,i+1}}{u_i}, G_{i-1,i+2},
                            \ldots, G_{i-1,\ell+1}),
    \end{equation*}
    where
    \begin{equation*}
            u_i = \begin{cases}
                     u       & \text{if } i \text{ odd} \\
                     \rev{u} & \text{otherwise}.
                  \end{cases}
    \end{equation*}
    For example, for $i=1$, $\graphseq_1$ is the graph sequence where $u$ is applied to the graph $G_{0,2}$.
    The remaining sequences $\graphseq_{\ell+1},\dots,\graphseq_{2\ell+1}$ remove the updates made by applying $\trans{G_1}{G_2}$.
    Let $(u_1', \ldots) := \trans{G_1}{G_2}$. For $i \in \{ \ell + 1, \ldots, 2\ell + 1 \}$, we define
    \begin{equation*}
            \graphseq_i := (G_{i-1,1}, \ldots, G_{i-1,i}, \apply{G_{i-1,i+1}}{\rev{u_i'}}, \apply{G_{i-1,i+2}}{\rev{u_i'}}, \ldots, \apply{G_{i-1,\ell}}{\rev{u_i'}}).
    \end{equation*}
    In other words, in $\graphseq_i$, the first $i$ updates of $\transseq{G_1}{G_2}$ are not present compared to $\graphseq_{\ell}$.
    Since $\transseq{G_1}{G_2}$ is a shortest sequence that transforms $G_1$ into $G_2$, an edge is either inserted or deleted at most once.
    Thus, reverting $u'_i$ by executing $\apply{}{\rev{u'_i}}$ for all graphs $G_{i-1,j}$ with $j \geq i+1$ cannot invalidate the graph sequence.
    The claim follows because $u'_i$ only affects a single edge.
\end{proof}
}
\opt{full}{\LowerBoundProof}

\begin{fact}
        Minimum spanning tree has spread $(\Theta(nW),\Theta(n))$. Minimum cut has spread $(\Theta(nW),\Theta(n^2))$. Maximal matching has spread $(\Theta(n),\Theta(n))$. Maximum cardinality matching has spread $(\Theta(n),\Theta(n))$. Maximum weight matching has spread $(\Theta(nW),\Theta(n))$.
\end{fact}

\bibliography{dynamic-privacy-paper.bib}

\begin{thebibliography}{10}

\bibitem{AroDif19}
Raman Arora and Jalaj Upadhyay.
\newblock On differentially private graph sparsification and applications.
\newblock In H.~Wallach, H.~Larochelle, A.~Beygelzimer, F.~{dAlch{\'e}-Buc}, E.~Fox, and R.~Garnett, editors, {\em Advances in Neural Information Processing Systems}, volume~32, pages 13399--13410. {Curran Associates, Inc.}, 2019.

\bibitem{BarPri07}
Boaz Barak, Kamalika Chaudhuri, Cynthia Dwork, Satyen Kale, Frank McSherry, and Kunal Talwar.
\newblock Privacy, accuracy, and consistency too: A holistic solution to contingency table release.
\newblock In {\em Proceedings of the Twenty-Sixth {{ACM SIGMOD}}-{{SIGACT}}-{{SIGART}} Symposium on {{Principles}} of Database Systems}, {{PODS}} '07, pages 273--282, {New York, NY, USA}, June 2007. {Association for Computing Machinery}.
\newblock \href {https://doi.org/10.1145/1265530.1265569} {\path{doi:10.1145/1265530.1265569}}.

\bibitem{BloDif13}
Jeremiah Blocki, Avrim Blum, Anupam Datta, and Or~Sheffet.
\newblock Differentially private data analysis of social networks via restricted sensitivity.
\newblock In {\em Proceedings of the 4th Conference on {{Innovations}} in {{Theoretical Computer Science}}}, {{ITCS}} '13, pages 87--96, {New York, NY, USA}, January 2013. {Association for Computing Machinery}.
\newblock \href {https://doi.org/10.1145/2422436.2422449} {\path{doi:10.1145/2422436.2422449}}.

\bibitem{BluLea13}
BlumAvrim, LigettKatrina, and RothAaron.
\newblock A learning theory approach to noninteractive database privacy.
\newblock {\em Journal of the ACM (JACM)}, May 2013.
\newblock \href {https://doi.org/10.1145/2450142.2450148} {\path{doi:10.1145/2450142.2450148}}.

\bibitem{ChaDif12}
T.~H.~Hubert Chan, Mingfei Li, Elaine Shi, and Wenchang Xu.
\newblock Differentially {{Private Continual Monitoring}} of {{Heavy Hitters}} from {{Distributed Streams}}.
\newblock In Simone {Fischer-H{\"u}bner} and Matthew Wright, editors, {\em Privacy {{Enhancing Technologies}}}, Lecture {{Notes}} in {{Computer Science}}, pages 140--159, {Berlin, Heidelberg}, 2012. {Springer}.
\newblock \href {https://doi.org/10.1007/978-3-642-31680-7_8} {\path{doi:10.1007/978-3-642-31680-7_8}}.

\bibitem{chan11}
T.-H.~Hubert Chan, Elaine Shi, and Dawn Song.
\newblock Private and continual release of statistics.
\newblock {\em ACM Trans. Inf. Syst. Secur.}, 14(3), November 2011.
\newblock \href {https://doi.org/10.1145/2043621.2043626} {\path{doi:10.1145/2043621.2043626}}.

\bibitem{CheRec13}
Shixi Chen and Shuigeng Zhou.
\newblock Recursive mechanism: Towards node differential privacy and unrestricted joins.
\newblock In {\em Proceedings of the 2013 {{ACM SIGMOD International Conference}} on {{Management}} of {{Data}}}, {{SIGMOD}} '13, pages 653--664, {New York, NY, USA}, June 2013. {Association for Computing Machinery}.
\newblock \href {https://doi.org/10.1145/2463676.2465304} {\path{doi:10.1145/2463676.2465304}}.

\bibitem{DayPub16}
Wei-Yen Day, Ninghui Li, and Min Lyu.
\newblock Publishing {{Graph Degree Distribution}} with {{Node Differential Privacy}}.
\newblock In {\em Proceedings of the 2016 {{International Conference}} on {{Management}} of {{Data}}}, {{SIGMOD}} '16, pages 123--138, {New York, NY, USA}, June 2016. {Association for Computing Machinery}.
\newblock \href {https://doi.org/10.1145/2882903.2926745} {\path{doi:10.1145/2882903.2926745}}.

\bibitem{DwoDif06}
Cynthia Dwork.
\newblock Differential {{Privacy}}.
\newblock In Michele Bugliesi, Bart Preneel, Vladimiro Sassone, and Ingo Wegener, editors, {\em Automata, {{Languages}} and {{Programming}}}, Lecture {{Notes}} in {{Computer Science}}, pages 1--12, {Berlin, Heidelberg}, 2006. {Springer}.
\newblock \href {https://doi.org/10.1007/11787006_1} {\path{doi:10.1007/11787006_1}}.

\bibitem{DwoDif09}
Cynthia Dwork and Jing Lei.
\newblock Differential privacy and robust statistics.
\newblock In {\em Proceedings of the Forty-First Annual {{ACM}} Symposium on {{Theory}} of Computing}, {{STOC}} '09, pages 371--380, {New York, NY, USA}, May 2009. {Association for Computing Machinery}.
\newblock \href {https://doi.org/10.1145/1536414.1536466} {\path{doi:10.1145/1536414.1536466}}.

\bibitem{DwoCal06}
Cynthia Dwork, Frank McSherry, Kobbi Nissim, and Adam Smith.
\newblock Calibrating {{Noise}} to {{Sensitivity}} in {{Private Data Analysis}}.
\newblock In Shai Halevi and Tal Rabin, editors, {\em Theory of {{Cryptography}}}, Lecture {{Notes}} in {{Computer Science}}, pages 265--284, {Berlin, Heidelberg}, 2006. {Springer}.
\newblock \href {https://doi.org/10.1007/11681878_14} {\path{doi:10.1007/11681878_14}}.

\bibitem{dwork10}
Cynthia Dwork, Moni Naor, Toniann Pitassi, and Guy~N. Rothblum.
\newblock Differential privacy under continual observation.
\newblock In {\em Proceedings of the Forty-Second ACM Symposium on Theory of Computing}, STOC '10, page 715–724, New York, NY, USA, 2010. Association for Computing Machinery.
\newblock \href {https://doi.org/10.1145/1806689.1806787} {\path{doi:10.1145/1806689.1806787}}.

\bibitem{DwoCom09}
Cynthia Dwork, Moni Naor, Omer Reingold, Guy~N. Rothblum, and Salil Vadhan.
\newblock On the complexity of differentially private data release: Efficient algorithms and hardness results.
\newblock In {\em Proceedings of the Forty-First Annual {{ACM}} Symposium on {{Theory}} of Computing}, {{STOC}} '09, pages 381--390, {New York, NY, USA}, May 2009. {Association for Computing Machinery}.
\newblock \href {https://doi.org/10.1145/1536414.1536467} {\path{doi:10.1145/1536414.1536467}}.

\bibitem{EliDif19}
Marek Eli{\'a}s, Michael Kapralov, Janardhan Kulkarni, and Yin~Tat Lee.
\newblock Differentially {{Private Release}} of {{Synthetic Graphs}}.
\newblock In {\em Proceedings of the 2020 {{ACM}}-{{SIAM Symposium}} on {{Discrete Algorithms}} ({{SODA}})}, Proceedings, pages 560--578. {Society for Industrial and Applied Mathematics}, December 2019.
\newblock \href {https://doi.org/10.1137/1.9781611975994.34} {\path{doi:10.1137/1.9781611975994.34}}.

\bibitem{ErdPri15}
M.~A. Erdogdu and N.~Fawaz.
\newblock Privacy-utility trade-off under continual observation.
\newblock In {\em 2015 {{IEEE International Symposium}} on {{Information Theory}} ({{ISIT}})}, pages 1801--1805, June 2015.
\newblock \href {https://doi.org/10.1109/ISIT.2015.7282766} {\path{doi:10.1109/ISIT.2015.7282766}}.

\bibitem{GupDif10}
Anupam Gupta, Katrina Ligett, Frank McSherry, Aaron Roth, and Kunal Talwar.
\newblock Differentially {{Private Combinatorial Optimization}}.
\newblock In {\em Proceedings of the 2010 {{Annual ACM}}-{{SIAM Symposium}} on {{Discrete Algorithms}}}, Proceedings, pages 1106--1125. {Society for Industrial and Applied Mathematics}, January 2010.
\newblock \href {https://doi.org/10.1137/1.9781611973075.90} {\path{doi:10.1137/1.9781611973075.90}}.

\bibitem{HarMul10}
Moritz Hardt and Guy~N. Rothblum.
\newblock A {{Multiplicative Weights Mechanism}} for {{Privacy}}-{{Preserving Data Analysis}}.
\newblock In {\em 2010 {{IEEE}} 51st {{Annual Symposium}} on {{Foundations}} of {{Computer Science}}}, pages 61--70, October 2010.
\newblock \href {https://doi.org/10.1109/FOCS.2010.85} {\path{doi:10.1109/FOCS.2010.85}}.

\bibitem{HayAcc09}
M.~Hay, C.~Li, G.~Miklau, and D.~Jensen.
\newblock Accurate {{Estimation}} of the {{Degree Distribution}} of {{Private Networks}}.
\newblock In {\em 2009 {{Ninth IEEE International Conference}} on {{Data Mining}}}, pages 169--178, December 2009.
\newblock \href {https://doi.org/10.1109/ICDM.2009.11} {\path{doi:10.1109/ICDM.2009.11}}.

\bibitem{KarPri11}
Vishesh Karwa, Sofya Raskhodnikova, Adam Smith, and Grigory Yaroslavtsev.
\newblock Private analysis of graph structure.
\newblock {\em Proceedings of the VLDB Endowment}, 4(11):1146--1157, August 2011.
\newblock \href {https://doi.org/10.14778/3402707.3402749} {\path{doi:10.14778/3402707.3402749}}.

\bibitem{KasWha08}
S.~P. Kasiviswanathan, H.~K. Lee, K.~Nissim, S.~Raskhodnikova, and A.~Smith.
\newblock What {{Can We Learn Privately}}?
\newblock In {\em 2008 49th {{Annual IEEE Symposium}} on {{Foundations}} of {{Computer Science}}}, pages 531--540, October 2008.
\newblock \href {https://doi.org/10.1109/FOCS.2008.27} {\path{doi:10.1109/FOCS.2008.27}}.

\bibitem{KasAna13}
Shiva~Prasad Kasiviswanathan, Kobbi Nissim, Sofya Raskhodnikova, and Adam Smith.
\newblock Analyzing {{Graphs}} with {{Node Differential Privacy}}.
\newblock In Amit Sahai, editor, {\em Theory of {{Cryptography}}}, Lecture {{Notes}} in {{Computer Science}}, pages 457--476, {Berlin, Heidelberg}, 2013. {Springer}.
\newblock \href {https://doi.org/10.1007/978-3-642-36594-2_26} {\path{doi:10.1007/978-3-642-36594-2_26}}.

\bibitem{KelDif14}
Georgios Kellaris, Stavros Papadopoulos, Xiaokui Xiao, and Dimitris Papadias.
\newblock Differentially private event sequences over infinite streams.
\newblock {\em Proceedings of the VLDB Endowment}, 7(12):1155--1166, August 2014.
\newblock \href {https://doi.org/10.14778/2732977.2732989} {\path{doi:10.14778/2732977.2732989}}.

\bibitem{LuExp14}
Wentian Lu and Gerome Miklau.
\newblock Exponential random graph estimation under differential privacy.
\newblock In {\em Proceedings of the 20th {{ACM SIGKDD}} International Conference on {{Knowledge}} Discovery and Data Mining}, {{KDD}} '14, pages 921--930, {New York, NY, USA}, August 2014. {Association for Computing Machinery}.
\newblock \href {https://doi.org/10.1145/2623330.2623683} {\path{doi:10.1145/2623330.2623683}}.

\bibitem{LyuUnd16}
Min Lyu, Dong Su, and Ninghui Li.
\newblock Understanding the sparse vector technique for differential privacy.
\newblock {\em Proceedings of the VLDB Endowment}, 10(6), 2017.

\bibitem{McSMec07}
F.~McSherry and K.~Talwar.
\newblock Mechanism {{Design}} via {{Differential Privacy}}.
\newblock In {\em 48th {{Annual IEEE Symposium}} on {{Foundations}} of {{Computer Science}} ({{FOCS}}'07)}, pages 94--103, October 2007.
\newblock \href {https://doi.org/10.1109/FOCS.2007.66} {\path{doi:10.1109/FOCS.2007.66}}.

\bibitem{NisSmo07}
Kobbi Nissim, Sofya Raskhodnikova, and Adam Smith.
\newblock Smooth sensitivity and sampling in private data analysis.
\newblock In {\em Proceedings of the Thirty-Ninth Annual {{ACM}} Symposium on {{Theory}} of Computing}, {{STOC}} '07, pages 75--84, {New York, NY, USA}, June 2007. {Association for Computing Machinery}.
\newblock \href {https://doi.org/10.1145/1250790.1250803} {\path{doi:10.1145/1250790.1250803}}.

\bibitem{NyDif14}
J.~Le Ny and G.~J. Pappas.
\newblock Differentially {{Private Filtering}}.
\newblock {\em IEEE Transactions on Automatic Control}, 59(2):341--354, February 2014.
\newblock \href {https://doi.org/10.1109/TAC.2013.2283096} {\path{doi:10.1109/TAC.2013.2283096}}.

\bibitem{RotInt10}
Aaron Roth and Tim Roughgarden.
\newblock Interactive privacy via the median mechanism.
\newblock In {\em Proceedings of the Forty-Second {{ACM}} Symposium on {{Theory}} of Computing}, {{STOC}} '10, pages 765--774, {New York, NY, USA}, June 2010. {Association for Computing Machinery}.
\newblock \href {https://doi.org/10.1145/1806689.1806794} {\path{doi:10.1145/1806689.1806794}}.

\bibitem{song18}
Shuang Song, Susan Little, Sanjay Mehta, Staal Vinterbo, and Kamalika Chaudhuri.
\newblock Differentially private continual release of graph statistics, 2018.
\newblock \href {https://arxiv.org/abs/1809.02575} {\path{arXiv:1809.02575}}.

\bibitem{WanRea18}
Q.~Wang, Y.~Zhang, X.~Lu, Z.~Wang, Z.~Qin, and K.~Ren.
\newblock Real-{{Time}} and {{Spatio}}-{{Temporal Crowd}}-{{Sourced Social Network Data Publishing}} with {{Differential Privacy}}.
\newblock {\em IEEE Transactions on Dependable and Secure Computing}, 15(4):591--606, July 2018.
\newblock \href {https://doi.org/10.1109/TDSC.2016.2599873} {\path{doi:10.1109/TDSC.2016.2599873}}.

\bibitem{WanDif13}
Yue Wang, Xintao Wu, and Leting Wu.
\newblock Differential {{Privacy Preserving Spectral Graph Analysis}}.
\newblock In Jian Pei, Vincent~S. Tseng, Longbing Cao, Hiroshi Motoda, and Guandong Xu, editors, {\em Advances in {{Knowledge Discovery}} and {{Data Mining}}}, Lecture {{Notes}} in {{Computer Science}}, pages 329--340, {Berlin, Heidelberg}, 2013. {Springer}.
\newblock \href {https://doi.org/10.1007/978-3-642-37456-2_28} {\path{doi:10.1007/978-3-642-37456-2_28}}.

\bibitem{ZhaPri15}
Jun Zhang, Graham Cormode, Cecilia~M. Procopiuc, Divesh Srivastava, and Xiaokui Xiao.
\newblock Private {{Release}} of {{Graph Statistics}} using {{Ladder Functions}}.
\newblock In {\em Proceedings of the 2015 {{ACM SIGMOD International Conference}} on {{Management}} of {{Data}}}, {{SIGMOD}} '15, pages 731--745, {New York, NY, USA}, May 2015. {Association for Computing Machinery}.
\newblock \href {https://doi.org/10.1145/2723372.2737785} {\path{doi:10.1145/2723372.2737785}}.

\bibitem{ZhaDif21}
Sen Zhang, Weiwei Ni, and Nan Fu.
\newblock Differentially {{Private Graph Publishing}} with {{Degree Distribution Preservation}}.
\newblock {\em Computers \& Security}, page 102285, April 2021.
\newblock \href {https://doi.org/10.1016/j.cose.2021.102285} {\path{doi:10.1016/j.cose.2021.102285}}.

\end{thebibliography}

\opt{full}{

\appendix
\usetag{isappendix}

\tagged{notused}{
\section{Update Sequences for Node Adjacency}
With \cref{def:node-adjacency} and the definition of graph sequences in \cref{sec:preliminaries},
two node-adjacent differ in all updates that involve an edge incident to $v^*$.
Here, we provide equivalent definitions that allow us to specify adjacent graph sequences that differ only in the insertion or deletion of node $v^*$, and have identical edge update sets.

A \emph{graph sequence} is a sequence of graphs $\graphseq = (G_1,G_2,\dots)$, where each $G_t = (V_t,E_t)$ is associated with time step $t$.
We write the length of a graph sequence as $|\graphseq|$.
A graph sequence is associated with an initial graph $G_0 = (V_0,E_0)$.
Each $V_t$ for $t \geq 0$ is a subset of the \emph{universe} $V^*$, which contains all nodes that can ever be present in a graph sequence.
For each time step $t \geq 1$, the insertion and deletion of nodes is described by the sets $\Vins{t} \subseteq V^*$ and $\Vdel{t} \subseteq V_{t-1}$.
We define $V_t = (V_{t-1} \setminus \Vdel{t}) \cup \Vins{t}$.
Insertion and deletion of edges is described by the sets $\Eins{t}$ and $\Edel{t}$.
We define $E_t = (E_{t-1} \setminus \Edel{t}) \cup \{(u,v) \in \Eins{t} \mid u,v \in V_t\}$.

Informally, at each time step we specify which nodes in the universe of nodes are present in the graph sequence.
When specifying a set of edges to insert into the graph, we can specify edges incident to nodes that are not present in the graph sequence.
These edges do not modify the graph sequence, but allows us to define adjacent graph sequences with identical edge update sets.

\begin{definition}[Node-adjacency]
        \label{def:node-adjacency-2}
        Let $\graphseq$, $\graphseq'$ be graph sequences as defined above with associated sequences of updates
        $(\Vdel{t})$, $(\Vins{t})$ and
        $(\Vdel{t}')$, $(\Vins{t}')$, and the edge update sets $\Eins{t}$, $\Edel{t}$ for both graph sequences.
        Assume w.l.o.g.\ that $\Vdel{t}' \subseteq \Vdel{t}$ and $\Vins{t}' \subseteq \Vins{t}$ for all $t$.
        $\graphseq$ and $\graphseq'$ are \emph{adjacent on $v^*$} if $|\graphseq| = |\graphseq'|$, there exists a node $v^*$
        and one of the following statements holds:
        \begin{enumerate}
                \item \label{item:node-adjacency-2-1}
                      $\Vdel{t} = \Vdel{t}'\,\forall\,t$ and 
                      $\exists t^*$ such that $\Vins{t} \setminus \Vins{t}' = \{v^*\}$ and
                      $\Vins{t} = \Vins{t}'\,\forall\,t\neq t^*$;
                \item $\Vins{t} = \Vins{t}'\,\forall\,t$ and 
                      $\exists t^*$ such that $\Vdel{t} \setminus \Vdel{t}' = \{v^*\}$ and
                      $\Vdel{t} = \Vdel{t}'\,\forall\,t\neq t^*$;
                \item $\exists t_1, t_2$ with $t_1 < t_2$ such that
                      $\Vins{t_1} \setminus \Vins{t_1}' = \{v^*\}$ and $\Vdel{t_2} \setminus \Vdel{t_2}' = \{v^*\}$ and
                      for all $t \neq t_1,t_2$ $\Vins{t} = \Vins{t}'$ and $\Vdel{t} = \Vdel{t'}$.
                \item $\exists t_1, t_2$ with $t_1 < t_2$ such that
                      $\Vdel{t_1} \setminus \Vdel{t_1}' = \{v^*\}$ and $\Vins{t_2} \setminus \Vins{t_2}' = \{v^*\}$ and
                      for all $t \neq t_1,t_2$ $\Vins{t} = \Vins{t}'$ and $\Vdel{t} = \Vdel{t'}$.
        \end{enumerate}
        Additionally, all edges in $\Eins{t}$ and $\Edel{t}$ are incident to at least one node in $\Vins{t}$ and $\Vdel{t}$, respectively.
\end{definition}
}

\opt{conf}{\section{Figures}\label{sec:fig-appendix}}
\opt{conf}{
        \CountingAlgorithm[H]
        \binmechfigure[H]
}

\opt{conf}{\section{Proofs Omitted from \cref{sec:global-sens-mech}}\label{sec:global-sens-proofs}}
\opt{conf}{
\subsection{Non-Binary Counting}
\label{sec:counting-proofs}
\errorcorollary
\obsonelemma*
\obsoneproof
\obsoneextlemma*
\obsoneextproof
\subsection{Graph Functions via Counting Mechanisms}
\label{sec:graph-mech-proofs}
\graphobsonelemma*
\graphobsoneproof
\binmechcorollary*
\binmechproof

\subsection{Bounds on Continuous Global Sensitivity}
\label{sec:gs-bounds-proofs}
\sensitivitytable
\GSproofs
\subsection{Minimum Spanning Tree Algorithms}
\label{sec:mst-algo-proofs}
\mstedgegslemma*
\mstedgegsproof
\mstnodegslemma*
\mstnodegsproof
}

\opt{conf}{\section{Proofs Omitted from \cref{sec:monotone}}\label{sec:monotone-proof}}
\opt{conf}{
\TechnicalMonotone*
\MonotoneProof
}

\opt{conf}{\section{Proofs Omitted from \cref{sec:edge-dp-lower-bounds}}\label{sec:edge-dp-lower-bounds-proofs}}
\opt{conf}{
\inclowerboundsection
}

\opt{conf}{\section{Proofs Omitted from \cref{sec:lower-bound-user-level}}\label{sec:lower-bound-proof}}
\opt{conf}{
        \begin{definition}
                \label{def:user-node-adj}
                Let $\graphseq = (G_1, \ldots), \graphseq' = (G'_1, \ldots)$ be graph sequences. The two sequences are node-adjacent \emph{on user-level} if there exists a node $v^*$ and a sequence of graph sequences $\mathcal{S} = (\graphseq_1, \ldots, \graphseq_\ell)$ so that $\graphseq_1 = \graphseq$, $\graphseq_\ell = \graphseq'$ and, for any $i \in [\ell-1]$, $\graphseq_i$ and $\graphseq_{i+1}$ are node-adjacent on $v^*$. An algorithm is $\eps$-edge-differentially private on user-level if it is $\eps$-differentially private when considering node-adjacency on user-level.
        \end{definition}

\begin{definition}[adjacency transformation]
        \label{def:transform}
        Let $G_1$ and $G_2$ be a pair of graphs. We define $\trans{G_1}{G_2}$ to be an update sequence $u_1, \ldots, u_\ell$ of minimum length that transforms $G_1$ into $G_2$.
We denote the graph sequence that results from applying $\trans{G_1}{G_2}$ to $G_1$ by $\transseq{G_1}{G_2}$, i.e., $\transseq{G_1}{G_2} = (G_1, \apply{G_1}{u_1}, \apply{(\apply{G_1}{u_1})}{u_2}, \ldots, G_2)$.
\end{definition}

\begin{definition}[spread]
        \label{def:spread}
        Let $s, \ell: \nats \rightarrow \{ 2i \mid i \in \nats \}$ be functions. A graph function $f$ has spread $(s(n), \ell(n))$ on inputs of size $n$ if, for every $n$, there exist two graphs $G_1, G_2$ of size $n$ so that $\lvert f(G_1) - f(G_2) \rvert \geq s(n)$
and $\transnum{G_1}{G_2} = \ell(n)$. Furthermore, $f$ \emph{spares an edge $e$} if $e \in E(G_1) \cap E(G_2)$.
\end{definition}

        \TechnicalLowerBound*
        \LowerBoundProof
}

}

\end{document}